\documentclass{theoretics}

\title{Constructing Deterministic Parity Automata from Positive and Negative Examples}
\ThCSauthor[aachen]{Le\'{o}n Bohn}{bohn@lics.rwth-aachen.de}[https://orcid.org/0000-0003-0881-3199]
\ThCSauthor[aachen]{Christof L\"{o}ding}{loeding@cs.rwth-aachen.de}[https://orcid.org/0000-0002-1529-2806]
\ThCSaffil[aachen]{RWTH Aachen University, Ahornstr. 55, 52074 Aachen, Germany}

\ThCSshortnames{L. Bohn and C. L\"{o}ding}
\ThCSshorttitle{Constructing Deterministic Parity Automata from Positive and Negative Examples}

\ThCSthanks{The first author was supported by DFG grant LO 1174/7-1.}

\ThCSkeywords{deterministic parity automata, learning from examples, learning in the limit, family of right congruences}

\ThCSyear{2024}
\ThCSarticlenum{17}
\ThCSreceived{Jun 21, 2023}
\ThCSrevised{Apr 3, 2024}
\ThCSaccepted{May 11, 2024}
\ThCSpublished{Jul 30, 2024}
\ThCSdoicreatedtrue

\usepackage{tikz}
\usepackage{stmaryrd}
\usepackage{xifthen}
\SetSymbolFont{stmry}{bold}{U}{stmry}{m}{n}
\DeclareFontFamily{U}{mathb}{\hyphenchar\font45}
\DeclareFontShape{U}{mathb}{m}{n}{
      <5> <6> <7> <8> <9> <10> gen * mathb
      <10.95> mathb10 <12> <14.4> <17.28> <20.74> <24.88> mathb12
      }{}
\DeclareSymbolFont{mathb}{U}{mathb}{m}{n}
\DeclareMathSymbol{\sqsubsetneq}{1}{mathb}{"88}

	\definecolor[named]{lipicsYellow}{rgb}{0.99,0.78,0.07} 
	\definecolor[named]{lipicsGray}{rgb}{0.31,0.31,0.33}

\usepackage{bm}
\usepackage[capitalise,noabbrev,nameinlink]{cleveref}

\usetikzlibrary{arrows,positioning,arrows.meta,shapes.geometric,shapes.misc,calc,automata,fit,backgrounds}
\tikzset{>=stealth, shorten >=1pt, auto}
\tikzset{automaton/.style={
			auto, shorten >=1pt, node distance = 10mm, >=stealth, transform shape, initial text=,
			initial distance=3mm,
			every edge/.append style={nodes={font=\small}},
			label/.append style={overlay, inner sep=1pt, minimum size=1mm},
			every state/.style={
					circle, text=black, fill=white,
					draw=black, inner sep=.1mm, minimum width=5mm, minimum height=3mm,
				}
		}}

\newcommand \green  [1] {{\color{teal} #1}}
\newcommand \blue  [1] {{\color{blue} #1}}
\newcommand \red [1] {{\color{red} #1}}

\definecolor{green}{RGB}{87,171,39}
\definecolor{blue} {RGB}{0,84,159}
\definecolor{orange}{RGB}{246,168,0}
\definecolor{red}{RGB}{204,7,30}

\newcommand{\eps}{\varepsilon}
\DeclareMathOperator{\prf}{\mathsf{Prf}}
\newcommand{\polyclass}{\ensuremath{d\text{-}\mathbb{IRC}(k\text{-}\mathbb{DPA})}\xspace}
\newcommand{\polyclassBuchi}{\ensuremath{d\text{-}\mathbb{IRC}(2\text{-}\mathbb{DPA})}\xspace}
\newcommand{\irc}{\ensuremath{\mathbb{IRC}(\mathbb{DPA})}\xspace}
\renewcommand{\iff}{\ensuremath{\Leftrightarrow}}

\newcommand{\true}{\ensuremath{\op{true}}\xspace}

\newcommand{\mealyal}{\ensuremath{\op{MMAL}}\xspace}
\newcommand{\dpainf}{\ensuremath{\op{DPAInf}}\xspace}
\newcommand{\glerc}{\ensuremath{\op{GLeRC}}\xspace}
\newcommand{\family}[1]{\overline{#1}}
\newcommand{\bow}[1]{{#1}_{\bowtie}}

\newcommand{\familyfor}[2]{\ensuremath{(#1_c)_{c \in [#2]}}}

\newcommand{\sem}[1]{#1}

\newcommand{\minrep}[1]{\ensuremath{\op{mr}(#1)}}

\newcommand{\classes}[1]{\ensuremath{[#1]}}

\newcommand{\topclass}[2]{%
	\ifthenelse{\isempty{#2}}{\op{T}_{#1}(\mathcal{F}_L)}{\op{T}_{#1}(#2)}}
\newcommand{\qinclass}[1]{\ensuremath{{Q_c}}}

\let\inf\relax\DeclareMathOperator{\inf}{\mathrm{Inf}}

\let\min\relax\DeclareMathOperator{\min}{\mathrm{min}}
\newcommand{\mc}[1]{\ensuremath{\mathcal{#1}}\xspace}

\newcommand{\FF}{\ensuremath{\mathcal{F}}}

\newcommand{\NNN}{\ensuremath{\mathbb{N}}}

\newcommand{\TT}{\ensuremath{\mathcal{T}}}
\newcommand{\HH}{\ensuremath{\mathcal{H}}}
\newcommand{\CC}{\ensuremath{\mathcal{C}}}

\newcommand{\DD}{\ensuremath{\mathcal{D}}}
\newcommand{\BB}{\ensuremath{\mathcal{B}}}

\newcommand{\MM}{{\ensuremath{\mathcal{M}}}}

\renewcommand{\AA}{{\ensuremath{\mathcal{A}}}}

\newcommand{\op}[1]{\ensuremath{\operatorname{\mathsf{#1}}}}
\renewcommand{\epsilon}{\varepsilon}

\newcommand{\Stay}[1]{\ensuremath{E_{#1}}}
\newcommand{\Amin}{\mc{A}_{\min}}
\newcommand{\refine}{\mathbin{\preceq}}

\begin{document}
\maketitle

\begin{abstract}
	We present a polynomial time algorithm that constructs a deterministic parity automaton (DPA) from a given set of positive and negative ultimately periodic example words.
	We show that this algorithm is complete for the class of $\omega$-regular languages, that is, it can learn a DPA for each regular $\omega$-language.
	For use in the algorithm, we give a definition of a DPA, that we call the precise DPA of a language, and show that it can be constructed from the syntactic family of right congruences for that language (introduced by Maler and Staiger in 1997). Depending on the structure of the language, the precise DPA can be of exponential size compared to a minimal DPA, but it can also be a minimal DPA. The upper bound that we obtain on the number of examples required for our algorithm to find a DPA for $L$ is therefore exponential in the size of a minimal DPA, in general. However, we identify two parameters of regular $\omega$-languages such that fixing these parameters makes the bound polynomial.
\end{abstract}

\section{Introduction} \label{intro}
Construction of deterministic finite automata (DFA) from labeled example words is usually referred to as passive learning of automata, and has been studied since the 1970s \cite{BiermannF72,TrakhtenbrotB73,Gold78}, see \cite{LopezG16} for a survey. A passive learning algorithm (passive learner for short) receives a sample $S = (S_+,S_-)$ of positive and negative example words as input, and should return a DFA that accepts all words in $S_+$ and rejects all words in $S_-$.  In order to find passive learners that are robust and generalize from the sample, Gold proposed the notion of ``learning in the limit'' \cite{learninginthelimit}. Given a class $\mc{C}$ of regular languages, a learner is said to learn every language in $\mc{C}$ in the limit, if for each language $L \in \mc{C}$ there is a DFA $\mathcal{A}$ and a characteristic sample $S^L$ that is consistent with $L$, such that the learner returns $\mathcal{A}$ for each sample that is consistent with $L$ and contains all examples from $S^L$. In this case, we also say that the learner is complete for the class $\mc{C}$.

In active DFA learning, the learning algorithm (referred to as active learner) has access to an oracle for a regular target language $L$. Angluin proposed an active learner that finds the minimal DFA for a target language $L$ in polynomial time based on membership and equivalence queries  \cite{activelearningangluin}. A membership query asks for a specific word if it is in the target language, and the oracle answers ``yes'' or ``no''. An equivalence query asks if a hypothesis DFA accepts the target language, and the oracle provides a  counterexample if it does not.

Starting from these first works on DFA learning, many variations of the basic algorithms have been developed and were implemented in recent years, e.g., in the framework learnlib \cite{IsbernerHS15} and the library flexfringe \cite{VerwerH17}.

The key property that is used in most learning algorithms for regular languages is the characterization of regular languages by the Myhill/Nerode congruence: For a language $L$, two words $u,v$ are equivalent if for each word $w$, either both $uw,vw$ are in $L$, or both are not in $L$. It is a basic result from automata theory that $L$ is regular if and only if the Myhill/Nerode congruence has finitely many classes, and that these classes can be used as state set of the minimal DFA for $L$ (see basic textbooks on automata theory, e.g.~\cite{HopcroftU69}). In particular, this means that two words can lead to the same state if they cannot be separated by a common suffix. Based on this observation, the transition structure of a minimal $n$ state DFA can be fully characterized by a set of examples (words together with the information whether they are in $L$ or not) whose size is polynomial in $n$.

In this paper, we consider the problem of constructing deterministic automata on infinite words, so called $\omega$-automata, from examples. Automata on infinite words have been studied since the early 1960s as a tool for solving decision problems in logic \cite{buchiOriginalPaper} (see also \cite{thomassurvey}), and are nowadays used in procedures for formal verification and synthesis of reactive systems (see, e.g., \cite{BaierK2008,Thomas09,MeyerSL18} for surveys and recent work). But in contrast to standard finite automata, very little is known about learning deterministic $\omega$-automata.

The first problem, how to represent infinite example words, is easily solved by considering ultimately periodic words. An infinite word is called periodic if it is of the form $v^\omega$ for a finite word $v$, and ultimately periodic if it is of the form $uv^\omega$ for finite words $u,v$ (where $v$ must be non-empty). It is a classical result that a regular $\omega$-language is uniquely determined by the ultimately periodic words that it contains, see e.g. \cite{buchiOriginalPaper} or \cite[Fact~1]{upwords}.

The main obstacle in learning deterministic $\omega$-automata is that there is no Myhill/Nerode-style characterization of deterministic $\omega$-automata.
It is still true that two finite words $u,v$ that are separated by $L$ with a common suffix  $w$ (which is an infinite word in this case) have to lead to different states in any deterministic $\omega$-automaton for $L$. But it is not true anymore that all words that are not separated can lead to the same state. Currently, there are no active or passive polynomial time learners that can infer a deterministic $\omega$-automaton for each regular $\omega$-language. The existing algorithms either learn different representations, use information about the target automaton, or can only learn subclasses of the regular $\omega$-languages (see related work at the end of this introduction).

In this paper, we focus on passive learning of deterministic parity automata (DPA).
A passive learner for DPAs receives a sample $S = (S_+,S_-)$ of positive and negative ultimately periodic example words as input, and returns a DPA that is consistent with the sample, that is, it accepts all words from $S_+$ and rejects all words from $S_-$. We are interested in such a passive learner that can learn a DPA for every regular $\omega$-language in the limit. Without further requirements, such an algorithm is fairly easy to obtain by simple enumeration: Iterate through all DPAs by increasing size and some lexicographic ordering for DPAs of same size, and output the first DPA that is consistent with $S$. This algorithm is easily seen to infer a smallest DPA for each regular $\omega$-language $L$ in the limit: A characteristic sample for $L$ only needs to contain example words that separate $L$ from each language $L'$ that is accepted by a DPA that precedes the first DPA for $L$ in the enumeration used by the algorithm. However, it is fairly obvious that the running time of such an algorithm is exponential. And it is known that already for DFAs, the minimum automaton identification problem is NP-hard \cite{Gold78} (follows also from \cite{Pfleeger73}). This easily transfers to DPAs.

So, as already proposed by Gold \cite{Gold78}, the time and data requirements of a passive learner are of interest. More precisely, the time complexity of a passive learner is its running time measured in the size of the input sample. A polynomial time passive learner constructs a DPA that is consistent with the input sample $S$ in time polynomial in the size of $S$. Note that this property is independent of the learning-in-the-limit property, so it is not related to the size of characteristic samples for identifying a language. This is what the data requirement is about: A passive learner for DPAs identifies a class $\mc{C}$ of regular $\omega$-languages from polynomial data if for each $L \in \mc{C}$ there is a characteristic sample $S^L$ for $L$ whose size is polynomial in the size of a smallest DPA for $L$. (Where, as for DFAs, $S^L$ is called characteristic for the passive learner and $L$ if the passive learner returns the same DPA for $L$ for each sample $S$ that is consistent with $L$ and contains all examples from $S^L$.)

Currently, no passive learner is known that runs in polynomial time and learns every regular $\omega$-language in the limit from polynomial data.
In fact, all passive learners for deterministic $\omega$-automata that have been presented so far only learn subclasses of the regular $\omega$-languages in the limit (see related work at the end of this introduction for more details).
We present, to the best of our knowledge, the first polynomial time passive learner for deterministic $\omega$-automata that can learn every regular $\omega$-language in the limit. Furthermore, although we can only show an exponential upper bound on the data requirement of our algorithm in general, the data requirement of our algorithm is polynomial for the subclasses of the regular $\omega$-languages for which learning in the limit from polynomial data is already known (see \cref{fixedparams}).

In more detail, our contributions can be summarized as follows:
\begin{enumerate}
	\item We introduce a DPA for a regular language $L$ that we call the precise DPA. The priority assignment computed by this DPA corresponds to a priority assignment that is obtained in a natural way by analyzing the periodic parts of words in the language. This analysis yields what we call the precise family of weak priority mapping (precise FWPM), which is a family of mappings that assign priorities to finite words. We give a construction that joins this family of mappings into a single one. The minimal Mealy machine computing this join mapping corresponds to the precise DPA.
	\item We show how the precise DPA for $L$ can be constructed from the syntactic family of right congruences (FORC) of $L$ \cite{syntacticcongruence}
	      by showing that a family of Mealy machines for the precise FWPM can be obtained in polynomial time by assigning priorities to the classes of the syntactic FORC.
	      From the Mealy machines for the precise FWPM one then obtains a construction of the precise DPA that is only exponential in the number of required priorities.
	      In particular, it is polynomial (in the size of the syntactic FORC) if the number of priorities is fixed. This improves the known upper bound for the construction of a DPA for a language from its syntactic FORC that was known from \cite{AngluinBF16}, which first builds a nondeterministic Büchi automaton whose size is polynomial in the size of the FORC, and then determinizes this Büchi automaton.
	      Because of the last step, this construction is exponential in the size of the FORC.

	      \label{intro:FORCtoDPA}
	\item We present a polynomial time passive learner for DPAs (an algorithm for constructing a consistent DPA from given sets of positive and negative examples of ultimately periodic words).
	      The algorithm can be seen as an extension of the algorithm from \cite{BohnL22}, which is for deterministic Büchi automata. However, this extension has completely new parts
	      that are based on the insights on the precise DPA of a language and its construction from the syntactic FORC. The main steps of the algorithm are:
	      \begin{enumerate}
		      \item  Infer a FORC from the examples using a state merging technique as in \cite{BohnL21,BohnL22}. For this purpose, we propose a generic algorithm \glerc for inferring right congruences from samples. This algorithm is parameterized by a consistency function that is invoked in order to check if a merge (an inserted transition) causes an inconsistency with the sample. This algorithm can be instantiated in different settings and allows us to give a uniform presentation of passive learners for different right congruences and their completeness proofs for the learning in the limit property.
		      \item  Compute a priority assignment on the prefixes of the example words, based on the construction of the precise DPA from the syntactic FORC mentioned in contribution~\ref{intro:FORCtoDPA}. We show that, although the construction of the precise DPA from the FORC is exponential, if it is only applied to the example words, then it induces a priority assignment of polynomial size. This is important because we cannot simply apply an exponential construction if we want to obtain  a polynomial time passive learner. 
		      \item Use (as black box) an active learning algorithm for Mealy machines (on finite words) to infer a DPA that is consistent with the sample, using the priority assignment from the previous step for answering queries of the active learning algorithm. This step extends the priority mapping that is computed in the previous step on the prefixes of the example words to the set of all finite words. Again, this step is done in order to bypass the direct construction of the DPA from the FORC, which is exponential.
	      \end{enumerate}
	      We show that this algorithm, in the limit, infers a DPA for each regular $\omega$-language $L$ (either the precise DPA or a smaller one). In general, our upper bound for the number of examples that is required for inferring a DPA for $L$ is exponential in the size of a minimal DPA for $L$. However, we identify two parameters of regular $\omega$-languages such that the required data is polynomial if we fix these parameters, generalizing the currently known results on passive learning of DPAs with polynomial time and data from \cite{AngluinFS20,BohnL21,BohnL22}.
\end{enumerate}
The paper is structured as follows. We continue this introduction with a discussion on related work. In \cref{preliminaries} we introduce required notation and results. In \cref{preciseDPA} we introduce the notion of precise DPA and prove some results on this class of DPAs. In \cref{forcs} we show how to construct the precise DPA of a language from its syntactic FORC. In \cref{dpalearner} we present the learning algorithm, and in \cref{conclusion} we conclude.

\subsection*{Related Work}
The first paper explicitly dealing with construction of $\omega$-automata from examples that we are aware of is \cite{HigueraJ04}, where example words $w$ are finite and can be of one of the following four types: all $\omega$-words with $w$ as prefix are in the language, at least one such $\omega$-word is in the language, all of them are outside the language, or at least one  is outside the language.
It is shown in \cite{HigueraJ04} that from such examples, an adaption of a state merging technique for DFA can learn all safety languages in the limit and runs in polynomial time on a sample (the class of safety languages corresponds to the regular $\omega$-languages that can be characterized by forbidden prefixes, and is quite restricted).

The construction of nondeterministic Büchi automata (NBAs) from examples is considered in \cite{BarthH12} by a reduction to SAT.
This construction is used in order to reduce the size of a given NBA.
Roughly speaking, the algorithm collects some ultimately periodic example words of the form $uv^\omega$ from the given NBA, and subsequently checks if there is a smaller NBA consistent with these words (using a SAT instance).
If it finds one, then it checks whether it is equivalent to the given NBA.
Otherwise, it adds a new example obtained from the equivalence test and continues.
Since the algorithm uses a reduction to SAT, it is clearly not a polynomial time algorithm.

The construction of deterministic parity automata from examples is considered in \cite{AngluinFS20} for the class of IRC(parity) languages, which are the $\omega$-languages that can be accepted by a parity automaton that uses the Myhill/Nerode congruence as its transitions structure.
The transition structure can hence be inferred from the examples in the same way as for DFA. It is shown in \cite{AngluinFS20} that the algorithm can infer every such language with polynomial time and data by using a decomposition of parity automata that is known from \cite{computingrabinindex} where it was used for the minimization of the number of required priorities.

The well-known RPNI algorithm \cite{rpniOG} that infers a DFA from examples by a state merging technique has been adapted to deterministic $\omega$-automata in \cite{BohnL21}, resulting in a polynomial time passive learner that can infer all IRC(parity) languages in the limit from polynomial data (the same is shown for other acceptance conditions like Büchi, generalized Büchi, and Rabin). The algorithm can also infer automata for languages that are not in this class, however it is also known that there are regular $\omega$-languages for which it cannot infer a correct automaton.

Finally, \cite{BohnL22} presents a polynomial time passive learner for deterministic Büchi automata (DBA) that can infer a DBA for every DBA-recognizable language in the limit. The best known upper bound for the number of examples and the size of the resulting DBA is, however, exponential in the size (which is the number of states) of a minimal DBA for the language. For the class of IRC(Büchi) languages, this algorithm only requires polynomial data.

There are also a few active learning algorithms for $\omega$-languages. We are focusing on passive learning here, but as shown in \cite[Proposition~13]{BohnL21}, a polynomial time active learner can be turned into a polynomial time passive learner through simulation, given that the class of target automata satisfies certain properties. We briefly summarize active learning algorithms for $\omega$-languages in this context. If not mentioned otherwise, the active learners use membership and equivalence queries.

The first active learning algorithm for $\omega$-languages can learn deterministic weak Büchi automata in polynomial time \cite{MalerP95}. This algorithm and the class of target automata satisfy the properties from \cite{BohnL21} and thus can be used to build a passive learner for deterministic weak Büchi automata, which define a subclass of IRC(Büchi) languages.

The first active learner for the full class of regular $\omega$-languages was proposed in \cite{FarzanCCTW08}. It can learn an NBA for each regular $\omega$-language $L$ by learning a representation of $L$ using finite words called $L_\$$ that was proposed in \cite{upwords}. This representation contains all finite words of the form $u\$v$ (for a fresh symbol $\$$) such that the ultimately periodic $\omega$-word $uv^\omega$ is in $L$. A DFA for this language can be learned using known active learning algorithms for DFAs. The main difficulty is that the oracle expects an $\omega$-language on equivalence queries, and that an intermediate DFA in the learning process might be inconsistent in the sense that it accepts $u\$v$ and rejects $u'\$v'$, although $uv^\omega = u'(v')^\omega$. The algorithm in \cite{FarzanCCTW08} transforms the DFAs into NBAs such that progress in the DFA learning algorithm can be ensured. The resulting active learner can learn an NBA for every regular $\omega$-language $L$ in time polynomial in the minimal DFA for $L_\$$ and in the length of the shortest counterexample returned by the teacher.
Note, however, that the size of a minimal DFA for $L_\$$ may be exponential in the size of a minimal NBA for $L$.

A similar idea is used in \cite{LiCZL17} that also presents an active learner for NBAs. Instead of $L_\$$, the algorithm learns a representation of the target $\omega$-language that is called family of DFAs (FDFA for short). This formalism has been introduced in \cite{Klarlund94} based on the notion of family of right congruences (FORC for short) from \cite{syntacticcongruence} (the technical report of \cite{syntacticcongruence} dates back to the same year as \cite{Klarlund94}). FDFAs are very similar to the $L_\$$-representation by DFAs. In FDFAs, the pairs are represented by several DFAs, one (actually just a deterministic transition system without accepting states) for the first component, called the leading automaton, and one DFA for each state $q$ of the leading DFA, called the progress automaton for $q$. A pair $(u,v)$ is accepted if $v$ is accepted from the progress DFA of state $q$ that is reached via $u$ in the leading automaton.
In contrast to a DFA for $L_\$$, an FDFA needs only to correctly accept/reject pairs $(u,v)$ for which $u$ and $uv$ reach the same state in the leading automaton (other pairs are don't cares).
For this reason, the (syntactic) FDFA can be exponentially more succinct than a minimal DFA for $L_\$$~\cite[Theorem~2]{AngluinF16} (the statement in \cite{AngluinF16} is for the periodic FDFA, which is almost the same as a minimal DFA for $L_\$$).

However, for FDFAs there is the same difficulty as for the $L_\$$ representation: An FDFA might accept some decompositions of an ultimately periodic word and reject others. An FDFA is called saturated if for each ultimately periodic word it accepts all relevant decompositions or rejects all of them.

The algorithm in \cite{LiCZL17} uses an active FDFA learning algorithm as presented in \cite{AngluinF16}.
But \cite{AngluinF16} requires a teacher that can take a hypothesis in form of a general FDFA (i.e. one that is not necessarily saturated)\footnote{The presentation of \cite{AngluinF16} contains an error because it assumes that the FDFAs used in the learning algorithm are saturated and hence define an $\omega$-language, which is not correct. So the algorithm only works with a teacher for FDFAs instead of $\omega$-languages.}.
For turning this into a learning algorithm with a teacher for regular $\omega$-languages, \cite{LiCZL17} proposes two alternative constructions for transforming FDFAs into NBAs.

Since these two active learners from \cite{FarzanCCTW08,LiCZL17} produce NBAs, they cannot be used to build a passive learner for a deterministic automaton model.

It is shown in \cite{AngluinBF16} that saturated FDFAs have many good closure and algorithmic properties similar to deterministic automata. We are not aware of any passive learners for saturated FDFAs. And the generic use of an active learner in a passive setting as described in \cite{BohnL21} does not work for the active FDFA learning algorithm from \cite{AngluinF16} because it works with general FDFAs, not only saturated ones, and the best known algorithm for checking whether an FDFA is saturated is in PSPACE \cite{AngluinBF16}. So one cannot build, based on the currently available results, a polynomial time passive learner for saturated FDFAs based on the active FDFA learner from \cite{AngluinF16}.

Finally, there is the active learning algorithm for deterministic parity automata presented in \cite{MichaliszynO20}. This algorithm does not purely use membership and equivalence queries, but additionally so called loop-index queries, which given an $\omega$-word $w$ return the number of symbols after which the run of the target automaton on $w$ enters the looping part. In order to answer such queries, one needs knowledge about the target automaton (and not just the language). So this active learner cannot be used to directly obtain a passive learner through simulation.

\section{Preliminaries}\label{preliminaries}
We use standard definitions and terminology from the theory of finite automata and $\omega$-automata, and assume some familiarity with these concepts (see, e.g., \cite{Thomas97,VardiW08,Wilke16} for some background).
An alphabet $\Sigma$ is a non-empty, finite set of symbols.
We use the standard notations $\Sigma^*,\Sigma^\omega$ for the sets of all finite words and all $\omega$-words, respectively, and let $\Sigma^+ := \Sigma^*\setminus\{\varepsilon\}$, where $\varepsilon$ is the empty word, and $\Sigma^\infty:=\Sigma^* \cup \Sigma^\omega$.
For $u \in \Sigma^*$ we write $|u|$ for its length.
A word $u \in \Sigma^*$ is called \emph{prefix} of $v \in \Sigma^\infty$ if $ux=v$ for some $x \in \Sigma^\infty$, and we write $u \sqsubseteq v$ in that case.
For a word $w \in \Sigma^\infty$, we use $\prf(w)$ to refer to the set of all prefixes of $w$, and for $X \subseteq\Sigma^\infty$, we write $\prf(X)$ for the union of all $\prf(x)$ for $x \in X$.
For $u \in \Sigma^*$ and $X \subseteq \Sigma^\infty$, we let $u^{-1}X := \{w \in \Sigma^\infty \mid uw \in X\}$.
We use the \emph{length-lexicographic (llex) ordering} on finite words that is based on some underlying ordering of the alphabet, and first compares words by length, and words of same length in the lexicographic ordering.

We call ${\sim} \subseteq \Sigma^* \times \Sigma^*$ a \emph{right congruence (RC)} if ${\sim}$ is an equivalence relation and  $u \mathrel{\sim} v$ implies $ua \mathrel{\sim} va$ for all $a \in \Sigma$. For a word $u \in \Sigma^*$, we use $[u]_{\sim}$ (just $[u]$ if ${\sim}$ is clear from the context) to denote the \emph{class} of $u$ in ${\sim}$ and write $\classes{{\sim}}$ for the set of all classes in ${\sim}$. We say that \emph{${\sim}_1$ refines ${\sim}_2$}, written as ${\sim}_1 \refine {\sim}_2$ if $u \mathrel{{\sim}_1} v$ implies $u \mathrel{{\sim}_2} v$.
We often use families indexed by classes of an RC. In such cases we also use words as indices representing their class. For example, if we have a family $(\gamma_c)_{c \in \classes{{\sim}}}$, then we often write $\gamma_u$ to refer to $\gamma_{[u]_{{\sim}}}$. The index or size, denoted $|{\sim}|$, of a right congruence ${\sim}$ is the number of its classes. In the following, we tacitly assume that all considered right congruences are of finite index, and not always explicitly mention this.

A (deterministic) transition system (TS) $\TT = (Q, \Sigma, \iota, \delta)$ over the finite alphabet $\Sigma$ consists of a finite, non-empty set of states $Q$, a transition function $\delta: Q \times \Sigma \to Q$ and an initial state $\iota \in Q$. We define $\delta^*: Q \times \Sigma^* \to Q$ inductively by  $\delta^*(q, \epsilon) = q$ and $\delta^*(q, aw) = \delta^*(\delta(q, a), w)$, and write $\delta^*(u)$ or $\TT(u)$ for $\delta^*(\iota,u)$. The \emph{run} of $\TT$ on $w = a_0a_1\ldots \in \Sigma^\infty$ from $q_0 \in Q$ is a (possibly infinite) sequence $q_0q_1\ldots$ with $q_{i+1} = \delta(q_i, a_i)$. We write $\inf_\TT(w)$ for the set of states that occur infinitely often in the run of $\TT$ on $w \in \Sigma^\omega$, and $\inf_\TT(X)$ for the union of $\inf_\TT(w)$ for $w \in X$.
A \emph{strongly connected component (SCC)} of $\TT$ is a maximal strongly connected set (with the usual definition). We write $\text{SCC}_\TT(q)$ for the SCC of $\TT$ that contains $q$.
A \emph{partial TS} is a TS in which $\delta$ is a partial function.

A right congruence ${\sim}$ induces a TS $\TT_{\sim} = ([{\sim}], \Sigma, [\varepsilon], \delta_{\sim})$, where $\delta_{\sim}(c, a) = [ca]$ with $[ca]$ being the class of $ua$ for an arbitrary $u \in c$. Vice versa, a transition system $\TT$ gives rise to the congruence ${\sim_\TT}$, where $u {\sim_\TT} v$ if and only if $\TT(u) = \TT(v)$. So we interchange these objects and often just speak of the TS ${\sim}$, meaning the TS $\TT_{\sim}$.

For a right congruence ${\sim}$ and a class $c \in \classes{{\sim}}$, we define $E^{\sim}_c = \{x \in \Sigma^+ \mid cx \mathrel{\sim} c\}$ as the set of non-empty words that loop on $c$. Often, we omit the superscript ${\sim}$ if it is clear from the context.

A \emph{language} (over $\Sigma$) is a subset of $\Sigma^*$.
A \emph{deterministic finite automaton (DFA)} $\DD$ consists of a transition system and a subset $F \subseteq Q$ of final or accepting states.
The accepted language is $L(\DD)$  with the standard definition, i.e.~the set of all words on which $\DD$ has a run from its initial state to some final state.
An \emph{$\omega$-language} is a set $L \subseteq \Sigma^\omega$ of $\omega$-words. There are different automaton models for defining the class of $\omega$-regular languages (see \cite{Thomas97,VardiW08,Wilke16}).
We are interested in (transition-based) deterministic parity automata defined below.

A \emph{priority mapping} is a function $\pi: \Sigma^+ \to \{0,\dotsc,k-1\}$ for some $k > 0$.
We overload notation and write $\pi(w)$ for $w \in \Sigma^\omega$ to denote the least $i$ such that $\pi(x) = i$ for infinitely many prefixes $x$ of $w$.
The language $L(\pi)$ defined by $\pi$ is the set of all $\omega$-words $w$ such that $\pi(w)$ is even.
We call $\pi$ \emph{weak} if $\pi(xy) \leq \pi(x)$ for all $x \in \Sigma^+$ and $y \in \Sigma^+$.
If ${\sim}$ is a right congruence and $\pi_c$ for each $c$ is a weak priority mapping, then $\family\pi = (\pi_c)_{c \in [{\sim}]}$ is called a \emph{family of weak priority mappings} (FWPM).

A \emph{deterministic parity automaton (DPA)} is of the form $\AA = (Q,\Sigma,\iota,\delta,\kappa)$ with a TS $\TT_\AA = (Q,\Sigma,\iota,\delta)$ and a function $\kappa:Q \times \Sigma \rightarrow \{0,\ldots,k-1\}$ mapping transitions in $\AA$ to priorities.
The size of \(\AA\) is the number of states, and we denote it by \(|\AA|\).
We overload notation and denote the priority mapping induced by a DPA as $\AA: \Sigma^+ \rightarrow \{0,\ldots,k-1\}$ which is defined as $\AA(ua) := \kappa(\delta^*(u),a)$ for $u \in \Sigma^*$ and $a \in \Sigma$.
The $\omega$-language $L(\AA)$ accepted by $\AA$ is the language of the induced priority mapping, i.e.~the set of all $\omega$-words $w$ such that the least priority seen infinitely often during the run of $\AA$ on $w$ is even.

For an automaton $\AA$, we use ${\sim}_\AA$ to refer to the right congruence of its transition system.
We call an $\omega$-language regular if it is accepted by a DPA.
The Myhill/Nerode congruence ${\sim}_L$ of a language $L \subseteq \Sigma^\circ$ for $\circ \in \{*,\omega\}$ is the right congruence defined by $u \mathrel{{\sim}_L} v$ if $\forall w \in \Sigma^\circ:\, (uw \in L \Leftrightarrow vw \in L)$.
If $L$ is regular, then ${\sim}_L$ is of finite index.
If ${\sim}$ is a right congruence that refines ${\sim}_L$ and $c \in [{\sim}]$ we let $L_c := u^{-1}L$, where $u$ is an arbitrary word in $c$.

For $u \in \Sigma^+$ and $q \in Q$ we use $\Amin(q,u)$ to refer to the minimal priority in the run of $\AA$ on $u$ from $q$.
For a right congruence ${\sim}$ that is refined by ${\sim}_\AA$, and a class $c$ of ${\sim}$, we call $q \in Q$ a $c$-state if the words that lead to $q$ in $\AA$ from the initial state are in $c$.

The \emph{parity complexity} of a regular $\omega$-language $L$, denoted $\op{pc}(L)$, is the least $k$ such that $L$ is accepted by a DPA with priorities $\{0,\ldots,k-1\}$ (one may distinguish also whether the smallest priority required is $0$ or $1$, but we omit this for simplicity).
We call a DPA $\AA = (Q,\Sigma,\iota,\delta,\kappa)$ \emph{normalized} if $\kappa$ is minimal in the following sense: For every $\kappa':Q \times \Sigma \rightarrow \NNN$ with $L((\TT_\AA,\kappa')) = L(\AA)$, we have $\kappa(q,a) \le \kappa'(q,a)$ for all $q \in Q$ and $a \in \Sigma$. This unique minimal priority function on $\TT_\AA$ can be computed in polynomial time \cite{computingrabinindex} (see also \cite{EhlersS22} for an adaption to transition-based DPAs). We refer to this as \emph{normalization of $\AA$}.

In some proofs in \cref{preciseDPA} we need some basic facts about normalized DPAs that are somehow known, but we cannot give a concrete reference. So we state and prove them in the following lemma.
The intuition behind the lemma is the following:
If a normalized DPA reads a word $u$ from some state $q$ and visits minimal priority $i$ on the way, then it is possible to complete a loop on $q$ that starts with $u$ and has $i$ as minimal priority (or the run changed the SCC and then $i=0$).
Furthermore, if it is possible to have a loop with priority $i$ on $q$, then also with priority $i-1$ (for $i \ge 2$).

\begin{lemma}\label{lem:normalizedDPA}
	Let $\AA = (Q,\Sigma,\iota,\delta,\kappa)$ be a normalized DPA, $q \in Q$ and $u \in \Sigma^+$ be such that $\Amin(q,u) = i$.
	\begin{enumerate}[(a)]
		\item Either there exists some word $v \in \Sigma^*$ with $\delta^*(q,uv) = q$ and $\Amin(q,uv) = i$, or $i=0$ and $\text{SCC}_\AA(\delta^*(q, u)) \cap \text{SCC}_\AA(q) = \emptyset$.
		\item If $i \ge 2$, then there is $y \in \Sigma^+$ with $\delta^*(q,y) = q$ and $\Amin(q,y) = i-1$.
	\end{enumerate}
\end{lemma}
\begin{proof}
	For (a), remove from $\AA$ all transitions with priority $< i$ and call the resulting partial DPA $\AA'$. Let $q' := \delta^*(q,u)$. Since $\Amin(q,u) = i$, the $u$-path from $q$ to $q'$ still exists in $\AA'$. If there is a path from $q'$ to $q$ in $\AA'$, then let $v$ be a word labeling such a path. Then $\Amin(q,uv) = i$ and $\delta^*(q,uv) = q$.

	If there is no path from $q'$ to $q$ in $\AA'$, let $u = u_1au_2$ with $u_1,u_2 \in \Sigma^*$ and $a \in \Sigma$ such that $q_1 := \delta^*(q,u_1)$ is in $\text{SCC}_{\AA'}(q)$ and $\delta(q_1,a)$ is not in $\text{SCC}_{\AA'}(q)$. Then every run in $\AA$ that takes the $a$-transition from $q_1$ infinitely often, contains a transition of priority $<i$ infinitely often. So, if $i>0$, we can change the priority of the $a$-transition from $q_1$ to $i-1$ without changing the language of $\AA$, contradicting the fact that $\AA$ is normalized (since the transition exists in $\AA'$, its priority is at least $i$ in $\AA$). Thus, $i=0$ and $\AA' = \AA$. This means that $q$ and $q'$ are in different SCCs and thus there is no path back from $q'$ to $q$.

	\medskip

	For (b), remove all transitions of priority $< i-1$ from $\AA$, obtaining $\AA''$. By (a), $\text{SCC}_{\AA''}(q)$ contains a loop and thus at least one transition.
	If $\text{SCC}_{\AA''}(q)$ contains a transition with priority $i-1$, we are done (take as $y$ the label of a path in $\AA''$ that starts in $q$, takes this ($i-1$)-transition, and goes back to $q$).

	Otherwise, we show that $\AA$ cannot be normalized: Lower in $\AA$ the priority of all transitions that are in $\text{SCC}_{\AA''}(q)$ by $2$. Since $i \ge 2$ and no transition in  $\text{SCC}_{\AA''}(q)$ has priority $i-1$, this is possible.
	Call the resulting DPA $\BB$. Consider the infinity set $X$ of transitions of a run (since $\AA$ and $\BB$ have the same transition structure, it is the same infinity set in both).
	If $X \cap \text{SCC}_{\AA''}(q) = \emptyset$, the smallest priority in $X$ is the same in $\AA$ and $\BB$.
	If $X \subseteq \text{SCC}_{\AA''}(q)$, then the minimal priority in $X$ differs by exactly $2$ in $\AA$ and $\BB$.
	Otherwise, $X$ contains a transition that exits $\text{SCC}_{\AA''}(q)$, and hence has priority $\le i-2$.
	Since all transitions inside $\text{SCC}_{\AA''}(q)$ have priority $\ge i-2$ in $\AA$ and $\BB$, the minimal priority in $X$ is the same in $\AA$ and $\BB$.
	So in all cases, $X$ is accepting in $\AA$ if and only if it is in $\BB$, contradicting the fact that $\AA$ is normalized.
\end{proof}

Deterministic parity automata can also be viewed as Mealy machines (with the priorities as output alphabet). Given a DPA $\AA$, there is thus a unique minimal DPA that induces the same priority mapping (when minimizing $\AA$ as Mealy machine). Active learning algorithms for DFAs (as described in the introduction) can be adapted to Mealy machines. In \cref{dpalearner} we use the fact that there are polynomial time active learning algorithms for Mealy machines that produce growing sequences of hypotheses (so their hypotheses are always bounded by the target Mealy machine).
See \cite{SteffenHM11} for a detailed introduction to Mealy machines, their minimization and active learning.

For $v \in \Sigma^+$ we write $v^\omega$ for the \emph{periodic} $\omega$-word $vvv\cdots$, and call an $\omega$-word \emph{ultimately periodic} if it is of the form $uv^\omega$ for $u \in \Sigma^*$ and $v \in \Sigma^+$. It follows from \cite{buchiOriginalPaper} that two regular $\omega$-languages are the same if they contain the same ultimately periodic words, see also \cite[Fact 1]{upwords}.
We use \emph{$\omega$-samples} of the form $S=(S_+,S_-)$, where $S_\sigma$ for $\sigma \in \{+,-\}$ is a finite set of ultimately periodic words. For simplicity, we do not explicitly distinguish ultimately periodic words and their representations, and $uv^\omega \in S_\sigma$ means that some representation of that word is in $S_\sigma$. We sometimes write $uv^\omega \in S$ for $uv^\omega \in S_+ \cup S_-$.
A sample $S$ is consistent with $L$ if $L \cap \sem{S_-} = \emptyset$ and $\sem{S_+} \subseteq L$.

A \emph{passive learner} (for DPAs) is a function $f$ that maps $\omega$-samples to DPAs.
$f$ is called a polynomial-time learner if $f$ can be computed in polynomial time.
A learner $f$ is \emph{consistent} if it constructs from each $\omega$-sample $S = (S_+, S_-)$ a DPA $\AA$ such that $L(\AA$) is consistent with $S$.
We say that $f$ can \emph{learn every regular $\omega$-language in the limit} if for each such language $L$ there is a \emph{characteristic sample} $S^L$ such that $L(f(S^L)) = L$ and $f(S^L) = f(S)$ is the same DPA for all samples $S$ that are consistent with $L$ and contain $S^L$.
For a class $\mc{C}$ of regular $\omega$-languages we say that $f$ can learn every language in $\mc{C}$ in the limit \emph{from polynomial data} if the characteristic samples for the languages in $\mc{C}$ are of polynomial size (in a smallest DPA for the corresponding language).

\section{Precise DPA of a language}\label{preciseDPA}
In this section we introduce the precise DPA for a regular $\omega$-language $L$ (\cref{def:preciseDPA}).
The definition we give is parameterized by a right congruence ${\sim}$ of finite index that refines ${\sim}_L$, meaning there is a precise DPA for $L$ and each such ${\sim}$.
The precise DPA of an $\omega$-language $L$ is then the one where ${\sim}$ is the same as ${\sim}_L$.
For defining the precise DPA, we first introduce families of weak priority mappings that capture the periodic part of the language $L$ with regard to ${\sim}$ (up to \cref{lem:precisecoloringmonotonic}).
Then we show how to combine them into a single priority mapping (\cref{def:join}, \cref{lem:joinlanguage}), and how to build a DPA for this combined priority mapping (\cref{lem:bowacomputesjoin}).
We finish the section with an observation on precise DPAs, \cref{lem:polynomialjoincoloring}, which later on plays a crucial role in achieving polynomial runtime of our learner.

To get an intuition for the family of priority mappings based on which we later define the precise DPA for a language, we fix a regular \(\omega\)-language \(L\) and a right congruence \(\sim\) that refines \(\sim_L\).
Our goal is to define a priority mapping for each \(\sim\)-class \(c\) which assigns priorities to all finite non-empty words, such that it correctly classifies all words \(v\) that loop on \(c\), in the sense that the priority assigned to \(v\) should be even if and only if \(v^\omega \in L_c\).
Additionally, we want the mapping to be weak, that is, non-increasing with growing length of the words. The intuition for this is that the priority of a finite word $v$ should reflect the information that is contained in $v$ with respect to periods that loop on $c$ and start with $v$.
In \cref{ex:pisets} below that illustrates the formal definition, the language $L$ contains all words with finitely many $b$ or infinitely many $aba$.
Clearly, if $v$ contains the pattern $aba$, then this is the ``maximal'' possible information: all periodic words that start with $v$ are in the language.
So the priority that is assigned to such $v$ is $0$.
If $v$ contains $b$ but not $aba$, then all periodic words that start with $v$ are either not in the language or contain $aba$.
We assign priority $1$ to such words, corresponding to the information that a periodic word starting with $v$ might be not in $L$, or it has a prefix that already contains the maximal information, i.e.~it is assigned priority $0$.
Finally, if $v$ does not contain $b$, then a periodic word that starts with $v$ either is in $L$ because it contains $b$ finitely often, or it has a prefix that is already assigned a priority $0$ or $1$.
We assign $2$ to such $v$.

More formally, to determine the priority for a word \(v \in \Sigma^+\) in the context of a class \(c\) of \(\sim\), we proceed inductively as follows (recall that $L_c$ denotes the residual of $L$ from class $c$, and $E_c$ denotes the set of all finite words that loop on $c$).
If there exists no \(x\) such that $vx \in E_c$, the word \(v\) is assigned priority \(0\) since no extension of \(v\) can ever loop on class \(c\) and hence the priority of \(v\) is irrelevant in the context of \(c\).
Thus, in the following, we only need to consider words \(v\) that actually have some extension \(x\) such that \(vx\) loops on \(c\), and we write \(v^{-1}E_c\) to denote the set of all such extensions.
Now if for all \(x \in v^{-1}E_c\) holds that \((vx)^\omega \in L_c\), then we assign \(0\) to \(v\) because regardless of the suffix we append to \(v\), the \(\omega\)-iteration is in \(L_c\).
Then, priority \(1\) is assigned to words \(v\) such that $v$ is not assigned $0$, and for all extensions \(x \in v^{-1}E_c\) holds that either \((vx)^\omega \notin L_c\), or $(vx)^\omega$ has a prefix that is assigned priority $0$.
From here on for priorities \(i > 1\), we proceed inductively until a priority has been identified for each word.

The definition we obtain from this procedure has some similarities with the definition of \emph{natural color of an infinite word} from \cite{EhlersS22}, however we define a priority for each finite word in the context of each class of ${\sim}$, whereas \cite{EhlersS22} defines one priority for each infinite word.
We will now give the formal definition of a family of sets \(P_{c,i}^\sim(L)\) indexed by a class \(c\) of \(\sim\) and a natural number \(i\).
The set \(P_{c,i}^\sim(L)\) contains all words which are assigned priority \(i\) or smaller in the priority mapping for class \(c\). The priority mapping itself, which assigns to $v$ the smallest $i$ such that $v$ is in $P_{c,i}$ is then extracted in \cref{def:preciseFWPM}.

\begin{definition}\label{def:paritydecomposition}
	Let $L$ be a regular $\omega$-language and ${\sim}$ be a right congruence that refines ${\sim}_L$. For each class $c$ of ${\sim}$, we define
	\begin{align*}\label{languagecoloring}
		P_{c,0}^{\sim}(L) & = \{v \in \Sigma^+ \mid \forall x \in v^{-1}E_c: \ (vx)^\omega \in L_c\} \\   		P_{c,i}^{\sim}(L) & = \{v \in \Sigma^+ \mid \forall x \in v^{-1}E_c: \ \big((vx)^\omega \in P^{\sim}_{c,i-1}(L)\Sigma^\omega \text{ or } ((vx)^\omega \in L_c \text{ iff } i \text{ is even})\big)\}.
	\end{align*}
\end{definition}

One easily verifies that for a fixed class \(c\), the sets \(P_{c,i}^\sim(L)\) form an inclusion chain where a larger index \(i\) leads to a larger set.
This is because the first condition in the definition of \(P_{c,i}^\sim(L)\) is trivially satisfied by all words belonging to \(P_{c,i-1}^\sim(L)\).
For a class \(c\), we can now define a priority mapping that assigns to a finite word \(v\) the least priority such that \(v \in P_{c,i}^\sim(L)\).
A formal definition is deferred to \cref{def:preciseFWPM} below, until after we have showed that the obtained inclusion chain is guaranteed to be finite whenever \(L\) is a regular \(\omega\)-language and \(\sim\) refines \(\sim_L\).

\begin{example}\label{ex:pisets}
	Let $L \subseteq \{a,b\}^\omega$ consist of all $\omega$-words with finitely many occurrences of $b$, or infinitely many occurrences of the infix $aba$, i.e.
	\[
		L = \{w \in \{a,b\}^\omega \mid \text{\(w\) contains \(b\) finitely often or \(w\) contains \(aba\) infinitely often}\}.
	\]
	As ${\sim}$ we take ${\sim}_L$, which has only one class (adding or removing finite prefixes is irrelevant for membership in $L$).
	This also implies that the universally quantified \(x\) in \cref{def:paritydecomposition} ranges over all finite words.
	So a word \(v\) belongs to \(P_{\varepsilon,0}^\sim(L)\) if for any possible extension \(x \in \Sigma^*\) holds that \((vx)^\omega \in L_\varepsilon\).
	This is satisfied precisely by those words that contain the infix \(aba\).

	Consider now a word \(v\) that contains the letter \(b\) and let \(x \in \Sigma^*\), then clearly \((vx)^\omega\) is guaranteed to contain infinitely many \(b\).
	If \((vx)^\omega\) also contains infinitely many \(aba\), then it certainly has a prefix in \(P_{\epsilon,0}^{\sim}(L)\), which means it satisfies the first condition.
	Otherwise, \((vx)^\omega \notin L_\varepsilon\) and the second condition is satisfied.
	So overall, \(P_{\epsilon,1}^{\sim}(L)\) consists of all words that have an occurrence of \(b\).

	Finally, \(P_{\epsilon,2}^{\sim}(L)\) consists of all finite, non-empty words over \(\Sigma\).
	Indeed, either \((vx)^\omega\) contains \(b\) infinitely often, in which case \((vx)^\omega\) has a prefix in \(P_{\epsilon,1}^{\sim}(L)\), or there are only finitely many occurrences of \(b\), implying \((vx)^\omega \in L_\varepsilon\).
	The languages $P_{\varepsilon,i}^{\sim}(L)$ for $i \leq 3$ are given in the second column of the table in \cref{fig:pilanguagesformal}.
\end{example}

\begin{table}[ht]
	\centering
	\begin{tabular}{c||c||c|c}
		$i$ & $P_{\varepsilon, i}^{{\sim}_L}(L)$ & $P_{\varepsilon, i}^{{\sim}_{L'}}(L')$          & $P_{a,i}^{{\sim}_{L'}}(L')$                     \\
		\hline
		$0$ & $\Sigma^*aba\Sigma^*$              & \multicolumn{2}{c}{$\Sigma^*aa\Sigma^*$}                                                          \\
		$1$ & $\Sigma^*b\Sigma^*$                & $\Sigma^*a\Sigma^*$                             & $\Sigma^*(a + b\Sigma^*d + d\Sigma^*b)\Sigma^*$ \\
		$2$ & $\Sigma^+$                         & $\Sigma^*(a + b\Sigma^*d + d\Sigma^*b)\Sigma^*$ & $\Sigma^+$                                      \\
		$3$ & $\Sigma^+$                         & $\Sigma^+$                                      & $\Sigma^+$
	\end{tabular}
	\caption{An overview over languages arising from \cref{def:paritydecomposition} for the languages $L$ and $L'$ from \cref{ex:pisets} and \cref{ex:twoclasses}, respectively.}
	\label{fig:pilanguagesformal}
\end{table}

\begin{example}\label{ex:twoclasses}
	Let $\Sigma = \{a,b,d\}$ and consider the language
	\begin{align*}
		{L'} = & \{w \in \Sigma^\omega \mid \text{there are infinitely many infixes $aa$ in $w$}\}                                         \\
		\cup   & \{w \in \Sigma^\omega \mid \text{$|w|_a$ is finite and $\big($even iff $w$ contains infinitely many $b$ and $d \big)$}\}.
	\end{align*}
	As ${\sim}$, we use ${\sim}_{L'}$, which has two classes, $[\varepsilon]$ and $[a]$, which are reached when reading words containing an even respectively odd number of $a$.
	A depiction of ${\sim}$ can be found in \cref{fig:twoclassexampleparts} alongside DFAs recognizing the sets $P_{\varepsilon,i}^{\sim}({L'})$ and $P_{a,i}^{{\sim}}(L')$ for $i \leq 3$.

	Regardless of the class \(c\) we consider, if a word \(v \in \Sigma^+\) contains an infix \(aa\), then for every possible \(x \in v^{-1}E_c\), we have \((vx)^\omega\in L'_c\) because any such \(\omega\)-iteration must contain infinitely many occurrences of \(aa\).
	Hence, the sets \(P_{\varepsilon,0}^{\sim}({L'})\) and \(P_{a,0}^{\sim}({L'})\) are equal and contain precisely those words that have an infix \(aa\).
	Similarly, we see that if \(v\) contains \(a\), then either \((vx)^\omega\) contains the infix \(aa\) infinitely often or \((vx)^\omega\) does not belong to \(L_c\).
	Thus, \(P_{\varepsilon,1}^{\sim}({L'})\) and \(P_{a,1}^{\sim}({L'})\) contain all words with an infix \(a\).

	Consider now a word \(v\) that contains both \(b\) and \(d\), but not \(a\).
	Such words appear at different indices depending on the class \(c\).
	We have \(v \in P_{a,1}^\sim(L')\) since for each \(x \in v^{-1}E_a\), the \(\omega\)-iteration \((vx)^\omega\) is not in \(L_a\) unless \((vx)^\omega\) contains \(aa\) infinitely often.
	However, \(v \notin P_{\eps,1}^\sim(L')\) since already \(v^\omega\) has no prefix in \(P_{\eps,1}^\sim(L')\) and \(v^\omega \in L_\eps\). So $v$ is in $P_{\eps,2}^\sim(L')$.
	Finally, the sets \(P_{\eps,3}^{\sim}({L'})\), \(P_{a,2}^{\sim}({L'})\) and \(P_{a,3}^{\sim}({L'})\) coincide and contain all non-empty finite words over \(\Sigma\).
	The sequences of sets for both classes are depicted in the last two columns of the table in \cref{fig:pilanguagesformal}.
\end{example}

\begin{figure}[ht]
	\begin{center}
	\includegraphics[width=\textwidth]{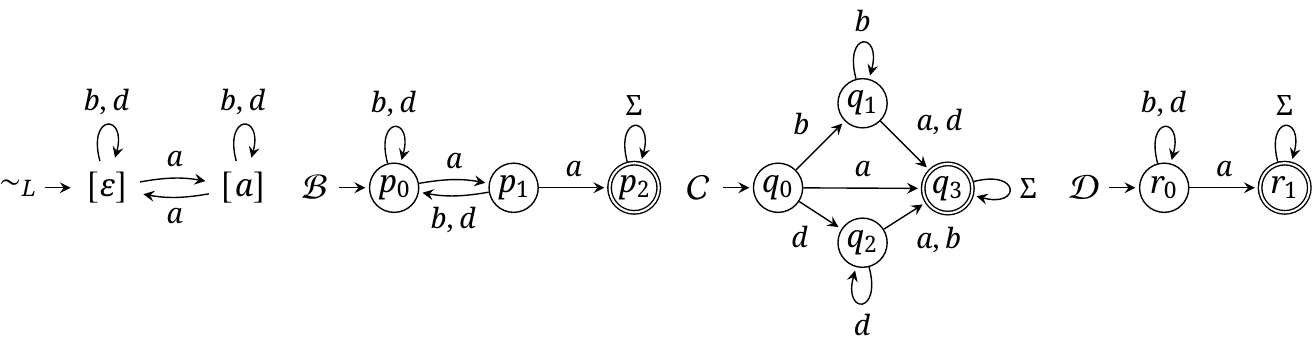}
	\end{center}
	\caption{The leading right congruence underlying the language $L$ from \cref{ex:twoclasses} is depicted on the far left.
	We have $L(\BB) = P_{0,\varepsilon}^{{\sim}_L}(L) = P_{0,a}^{{\sim}_L}(L)$, $L(\CC) = P_{2,\varepsilon}^{{\sim}_L}(L) = P_{1,a}^{{\sim}_L}(L)$ and $L(\DD) = P_{1,\varepsilon}^{{\sim}_L}(L)$.
	We omit the depiction of a DFA for $P_{3,\varepsilon}^{{\sim}_L}(L) = P_{2,a}^{{\sim}_L}(L) = \Sigma^+$.}
	\label{fig:twoclassexampleparts}
\end{figure}

Note that all sets $P_{c,i}$ arising from the given examples are regular. In \cref{ex:pisets} one can directly see that the indices $0,1,2$ naturally correspond to the priorities that a DPA would use to accept the language (emit priority $0$ whenever an infix $aba$ is detected, and priority $1$ for every occurrence of $b$, and priority $2$ otherwise). We show that this is not a coincidence by establishing a tight connection between the sets $P^{\sim}_{c,i}(L)$ and the priorities visited by a normalized DPA $\AA$ for $L$, namely,  $P^{\sim}_{c,i}(L)$ contains precisely those words $u$ that are guaranteed to take a transition of priority at most $i$ from every state in class $c$ (provided that ${\sim}_\AA$ refines ${\sim}$ as otherwise it makes no sense to speak of $c$-states).

\begin{lemma}\label{lem:normalizedDPAconnection}
	Let $\AA$ be a normalized DPA with priorities $\{0,\dotsc,k-1\}$ recognizing the language $L \subseteq \Sigma^\omega$ and ${\sim}$ be a right congruence with ${\sim}_\AA \refine {\sim} \refine {\sim}_L$. Then for each class $c$ of ${\sim}$ and each $i \ge 0$ holds $P^{\sim}_{c,i}(L) = P^{\sim}_{c,i}(\AA) := \{u \in \Sigma^+ \mid \max\{\Amin(q, u) \mid q \text{ is $c$-state}\} \leq i\}.$
\end{lemma}
\begin{proof}
	For every class $c$ of ${\sim}$ we show $P^{\sim}_{c,i}(L) = P^{\sim}_{c,i}(\AA)$ by induction on $i$. Note that by setting $P_{c,-1}^{\sim}(L)  = \emptyset$, we obtain a uniform definition of all $P_{c,i}^{\sim}(L)$ for $i \ge 0$. We can start the induction at $-1$. Since $\AA$ only uses priorities $\ge 0$, we obtain $ P^{\sim}_{c,-1}(\AA) = \emptyset$. So let $i \ge 0$ and assume by induction that $P_{c,i-1}^{\sim}(L) = P_{c,i-1}^{\sim}(\AA)$.

	We first show $P_{c,i}^{\sim}(L) \subseteq P_{c,i}^{\sim}(\AA)$ by establishing that for all $u \in P_{c,i}^{\sim}(L) \setminus P_{c,i-1}^{\sim}(L)$ holds $u \in P_{c,i}^{\sim}(\AA)$. Let $q$ be a $c$-state, and let $i' := \Amin(q,u)$. We have to show that $i' \le i$, so assume towards a contradiction that $i' > i$. Let $x$ be such that $\AA:q \xrightarrow{ux} q$ with minimal priority $i'$ on this loop (according to \cref{lem:normalizedDPA}).
	Then no prefix of $(ux)^\omega$ is in $P_{c,i}^{\sim}(\AA)$ because the minimal priority seen from $q$ on any prefix of  $(ux)^\omega$ is a least $i' > i$. Since $P_{c,i}^{\sim}(\AA) \supseteq P_{c,i-1}^{\sim}(\AA)$, also no prefix of $(ux)^\omega$ is in $P_{c,i-1}^{\sim}(\AA)$. By induction, it follows that no prefix of $(ux)^\omega$ is in $P_{c,i-1}^{\sim}(L)$, and thus $(ux)^\omega \in L_c$ if and only if $i$ is even. Since $(ux)^\omega$ is accepted from the $c$-state $q$ if and only if $i'$ is even, we obtain that $i$ is even if and only if $i'$ is even, and hence $i'-1 > i$ because $i' > i$.

	Let $y \in \Sigma^*$ be such that $\AA:q \xrightarrow{y} q$ with minimal priority $i'-1$ (according to \cref{lem:normalizedDPA}). Then $(uxy)^\omega$ is accepted from $q$ if and only if $i'-1$ is even if and only if $i$ is odd. Since $u$ is in $P_{c,i}^{\sim}(L)$, we obtain that $(uxy)^\omega$ has a prefix in $P_{c,i-1}^{\sim}(L)$. This implies by induction that $(uxy)^\omega$ has a prefix in $P_{c,i-1}^{\sim}(\AA)$, contradicting the fact that the minimal priority that is seen from $q$ on $(uxy)^\omega$ is $i'-1 > i$.

	To establish the other inclusion $P_{c,i}^{\sim}(\AA) \subseteq P_{c,i}^{\sim}(L)$, we show for all $u \in P_{c,i}^{\sim}(\AA) \setminus P_{c,i-1}^{\sim}(\AA)$ that $u \in P_{c,i}^{\sim}(L)$. So let $x \in \Sigma^*$ with $ux \in E_c$, and let $n > |\AA|$. Since $u \in P_{c,i}^{\sim}(\AA)$, also $(ux)^n \in P_{c,i}^{\sim}(\AA)$.

	If $(ux)^n \in P_{c,i-1}^{\sim}(\AA)$, then by induction $(ux)^\omega$ has a prefix in $P_{c,i-1}^{\sim}(L)$.

	If $(ux)^n \in P_{c,i}^{\sim}(\AA) \setminus P_{c,i-1}^{\sim}(\AA)$, then let $q$ be a $c$-state such that
	$\Amin(q, (ux)^n) = i$. Since $n > |\AA|$ and $ux \in E_c$, there is a $c$-state $q'$ and $n_1,n_2$ with $n_1+n_2 \le n$ such that $\AA:q \xrightarrow{(ux)^{n_1}} q'  \xrightarrow{(ux)^{n_2}} q'$. Since $u \in P_{c,i}^{\sim}(\AA)$, the minimal priority on the loop $q'  \xrightarrow{(ux)^{n_2}} q'$ is at most $i$. And by choice of $q$, the minimal priority on $\AA:q \xrightarrow{(ux)^{n_1}} q'  \xrightarrow{(ux)^{n_2}} q'$ is $i$. Hence, the minimal priority on the loop $q'  \xrightarrow{(ux)^{n_2}} q'$ is $i$, and thus $(ux)^\omega \in L_c$ if and only if $i$ is even.

	We conclude that $u \in P_{c,i}^{\sim}(L)$.
\end{proof}

We already noted above that the sets \(P_{c,i}^\sim(L)\) form an inclusion chain.
As consequence of \cref{lem:normalizedDPAconnection}, we obtain that this inclusion chain is of length \(k\) whenever \(L\) is a regular \(\omega\)-language of parity complexity \(k\) (the minimal number of priorities a DPA for \(L\) needs is \(k\)).
So $P_{c,0}^{\sim}(L) \subseteq \cdots \subseteq P_{c,k-1}^{\sim}(L) = \Sigma^+$ and each set along the chain is regular.
This implies that the priority mappings we obtain from the following definition are totally defined, weak, and can be computed by Mealy machines.

\begin{definition} \label{def:preciseFWPM}
	With the terminology from \cref{def:paritydecomposition} and $L$ of parity complexity $k$, we represent the sets $P_{c,i}^{\sim}(L)$ for each class by a priority mapping $\pi^{\sim}_c: \Sigma^+ \rightarrow \{0,\ldots,k-1\}$ defined by $\pi^{\sim}_c(u) = i$ for the minimal $i$ such that  $u \in P_{c,i}^{\sim}(L)$. We write $\family{\pi^{\sim}} = \familyfor{\pi^{\sim}}{{\sim}}$ for the family of these priority mappings, and refer to it as the \emph{precise FWPM for $L$ and ${\sim}$}.
\end{definition}
With $L$ and ${\sim}$ as in \cref{ex:pisets}, we obtain that $\pi_\varepsilon^{\sim}(u) = 0$ if $u$ contains $aba$, $\pi_\varepsilon^{\sim}(u) = 1$ if $u$ contains $b$ but not $aba$, and $\pi_\varepsilon^{\sim}(u) = 2$ otherwise. The left-hand side of \cref{fig:FWPM-Mealy} shows a Mealy machine that computes $\pi_\varepsilon^{\sim}(u)$.
We give a bound on the size of Mealy machines for computing the precise FWPM, which can be derived from \cref{lem:normalizedDPAconnection}.

\begin{figure}
\includegraphics[width=\textwidth]{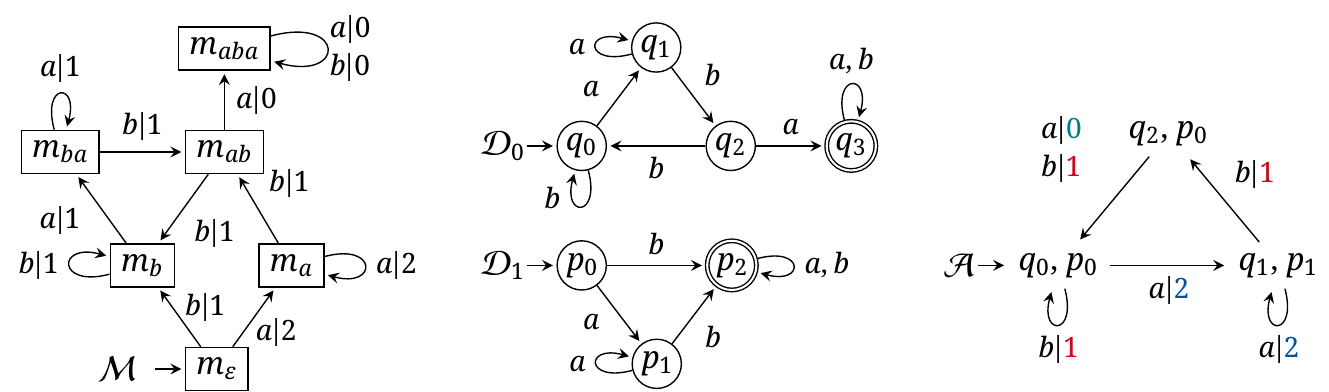}
	\caption{On the left-hand side a Mealy machine $\MM$ that computes $\pi_\varepsilon^{\sim}$ for $L$ and ${\sim}$ from \cref{ex:pisets} is shown. The DFAs $\DD_i$ for $i \in \{0,1\}$ in the middle accept precisely the words that are assigned priority $\le i$ by $\MM$. On the right-hand side the precise DPA that is obtained from the $\DD_i$ is depicted.}\label{fig:FWPM-Mealy}
\end{figure}

\begin{theorem}\label{mealysize}
	Let $\AA$ be a normalized DPA with $k$ priorities that recognizes $L$, and ${\sim}$ be a right congruence with ${\sim}_\AA \refine {\sim} \refine {\sim}_L$. Then for each class $c$ of ${\sim}$, the mapping $\pi_c^{\sim}$ can be computed by a Mealy machine $\MM_c$ with at most $m(dk)^d$ states, where $m$ is the number of classes of ${\sim}$, and each class of ${\sim}$-equivalent states has at most $d$ many states in $\AA$.
\end{theorem}
\begin{proof}
	According to \cref{lem:normalizedDPAconnection}, a Mealy machine only needs to keep track of the minimal priority that is visited from each $c$-state of $\AA$ (and it then outputs the maximum of those numbers).

	Let $\AA = (Q,\Sigma,q_0,\delta,\kappa)$ and $Q_c$ be the set of $c$-states in $\AA$. The states of $\MM_c$ are pairs $(\tau,\mu)$ of functions $\tau:Q_c \rightarrow Q$ and $\mu:Q_c \rightarrow \{0, \ldots,k-1\}$, where $\tau(q_1) \mathrel{\sim} \tau(q_2)$ for all $q_1,q_2 \in Q_c$. There are $m(dk)^d$ such pairs of functions (not all of which need to be reachable in $\MM$).

	The initial state is $(\tau_0,\mu_0)$ with $\tau_0(q) = q$ and $\mu_0(q) = k-1$ for all $q \in Q_c$. For $(\tau,\mu)$ and an input $a \in \Sigma$, let the $a$-successor of $(\tau,\mu)$ be $(\tau',\mu')$ with $\tau'(q) = \delta(\tau(q),a)$ and $\mu'(q) = \min(\mu(q),\kappa(\tau(q),a))$. The output of $\MM_c$ on this transition is $\max_{q \in Q_c}(\mu'(q))$.

	Then $\MM_c$ outputs on $u$ the least $i$ such that $u \in P_{c,i}(\AA)$, which corresponds to $\pi_c^{\sim}(u)$ according to \cref{lem:normalizedDPAconnection} and \cref{def:preciseFWPM}.
\end{proof}

The following definition expresses what it means for an FWPM to correctly capture the words that loop on the classes of the underlying right congruence \(\sim\).

\begin{definition} \label{def:capture-periodic}
	We say (with the terminology from \cref{def:paritydecomposition}) that an FWPM $(\gamma_c)_{c \in [{\sim}]}$ \emph{captures the periodic part of $(L, {\sim})$}, if for all classes $c \in [{\sim}]$ and  $v \in \Stay c$ holds that $v^\omega \in L_c$ if and only if $v^\omega \in L(\gamma_c)$.
\end{definition}

\begin{lemma}\label{lem:capperiodicpart}
	Let $L$ be a regular $\omega$-language and ${\sim}$ refine ${\sim}_L$.
	The family $\family{\pi^{\sim}} = (\pi_c^{\sim})_{c \in [{\sim}]}$ captures the periodic part of $(L,{\sim})$.
\end{lemma}
\begin{proof}
	Let $c$ be a class of ${\sim}$, $v \in \Stay c$ and set $i := \pi_c^{\sim}(v^\omega)$.
	We show that $i$ is even if and only if $v^\omega \in L_c$.
	Pick an $n$ such that $\pi_c^{\sim}(v^nx) = i$ for all $x \sqsubseteq v^\omega$.
	This is possible since $\pi_c^{\sim}$ is weak.
	By definition, $\pi_c^{\sim}(v^n) = i$ implies that either $(v^n)^\omega$ has a prefix in $P_{c,i-1}^{\sim}(L)$ or ($(v^n)^\omega \in L_c$ if and only if $i$ is even).
	In the former case, there would exist some $x \sqsubseteq v^\omega$ with $\pi_c^{\sim}(v^nx) < i$, contradicting our choice of \(n\).
	Hence, the second case must be true, which is what we set out to show.
\end{proof}

The following technical lemma is only used to establish the result in \cref{lem:leastFWPM} that gives justification for the name ``precise FWPM''.
As it is not required for the other results of this paper, readers may safely skip to after \cref{lem:leastFWPM}.

\begin{lemma}\label{lem:technical}
	Let $\family{\pi^{\sim}}$ be as in \cref{def:preciseFWPM}, \(c\) be a class of \(\sim\) and $v\in \Sigma^+$.
	If $\pi_c^{\sim}(v) = i \geq 1$, then there exists an $x \in \Sigma^*$ such that $vx \in \Stay c$ and $\pi_c^{\sim}((vx)^\omega) = i$.
	Moreover, if $\pi_c^\sim(v) = i \geq 2$, then there exists $y \in \Sigma^*$ with $vy \in \Stay c$ and $\pi_c^{\sim}((vy)^\omega) = i-1$.
\end{lemma}
\begin{proof}
	We show the first part of the claim.
	Let \(c\) be a class of \(\sim\), \(v \in \Sigma^+\) and \(\pi^\sim_c(v) = i \geq 1\).
	Since $v \notin P^{\sim}_{c,i-1}(L)$, there exists an $x \in \Sigma^*$ such that $vx \in \Stay c$, ($(vx)^\omega \in L_c$ if and only if $i-1$ is odd), and $\prf((vx)^\omega) \cap P^{\sim}_{c,i-2}(L) = \emptyset$ (this corresponds to the negation of the condition for membership in $P^{\sim}_{c,i-1}(L)$).
	Further, $(vx)^\omega$ also has no prefix in $P^{\sim}_{c,i-1}(L)$.
	Otherwise, $\pi^{\sim}_c((vx)^\omega) = i-1$, contradicting the fact that by \cref{lem:capperiodicpart}, $\family{\pi^{\sim}}$ captures the periodic part of $L$.
	Thus, it must be that $\pi_c^{\sim}((vx)^\omega) = i$.\\
	For the second part, assume towards a contradiction that $i \geq 2$ and $\pi_c^{\sim}((vy)^\omega) \neq i-1$ for all $y \in \Sigma^*$ with $vy \in \Stay c$.
	Then for each such $y$, either $\pi_c^{\sim}((vy)^\omega) \in \{i, i-2\}$ or $\pi_c^{\sim}((vy)^\omega) < i-2$.
	The former implies $(vy)^\omega \in L_c$ if and only if $i-2$ is even, while the latter means that $(vy)^\omega$ has some prefix in $P^{\sim}_{c,i-3}(L)$.
	In either case, $v \in P^{\sim}_{c,i-2}(L)$, contradicting $\pi_c^{\sim}(v) = i$.
\end{proof}

We now give the result that can be seen as justification for the name ``precise FWPM''.
It shows that in a certain sense, each mapping \(\pi_c^\sim\) of the precise FWPM provides the most accurate information with regard to words that loop on \(c\).
For this, we partially order FWPMs by point-wise comparison, that is, for two FWPMs $\family\gamma = \familyfor{\gamma}{{\sim}}$ and  $\family{\gamma}' = \familyfor{\gamma'}{{\sim}}$ we let $\family\gamma \leq \family{\gamma}'$ if for all classes $c$ of ${\sim}$ and $u \in \Sigma^+$ holds $\gamma_c(u) \leq \gamma_c'(u)$.
The following lemma now establishes that with regard to this order, no smaller FWPM than the precise FPWM correctly captures the periodic part of \((L, \sim)\).

\begin{lemma}\label{lem:leastFWPM}
	For a regular \(\omega\)-language \(L\) and a right congruence \(\sim\) that refines \(\sim_L\), the family \(\family{\pi^{\sim}} = (\pi_c^{\sim})_{c \in [{\sim}]}\) is the unique least FWPM that captures the periodic part of \((L, \sim)\).
\end{lemma}
\begin{proof}
	Let $\overline\gamma$ be an FWPM which captures the periodic part of $(L,{\sim})$. We show $\overline{\pi^{\sim}} \leq \overline\gamma$.
	Let $v \in \Sigma^+$ be an arbitrary word with $\gamma_c(v) = i$. We show by induction on $i$ that $\pi^{\sim}_c(v) \leq i$.
	For the base, let $i = 0$. Then $\gamma_c((vx)^\omega) = 0$ for all $x \in v^{-1}\Stay c$ since $\gamma_c$ is weak.
	This means for all $x \in \Sigma^*$, if $vx \in \Stay c$, then $(vx)^\omega \in L_c$ since $\overline\gamma$ captures the periodic part of $(L,{\sim})$.
	Thus, $v \in P_{c,0}^{\sim}$ and therefore $\pi^{\sim}_c(v) = 0$.

	For an $i > 0$, we assume towards a contradiction that $\pi_c^{\sim}(v) = j > i$, implying $j \geq 2$.
	By \cref{lem:technical}, there exists some $x \in \Sigma^*$ such that $vx \in \Stay c$ and $\pi_c^{\sim}((vx)^\omega) = j$.
	Further, it must be that $\gamma_c((vx)^\omega) = i$, as otherwise if $\gamma_c((vx)^\omega) < i$, it would follow by induction that $\pi^{\sim}_c((vx)^\omega) < i$.
	This means $j$ is even if and only if $i$ is even and thus $j - 1 > i$.\\
	Using \cref{lem:technical} once more, we pick some $y \in \Sigma^*$ with $vxy \in \Stay c$ and $\pi_c^{\sim}((vxy)^\omega) = j-1$.
	Again, by induction $\gamma_c((vxy)^\omega) = i$ and hence $(vxy)^\omega \in L(\gamma_c)$ if and only if $(vxy)^\omega \notin L(\pi^{\sim}_c)$.
	This contradicts our assumption that both $\overline\gamma$ and $\family{\pi^{\sim}}$ capture the periodic part of $(L, {\sim})$.
	Hence, $\pi_c^{\sim}(v) \leq i$, concluding the proof.
\end{proof}

The precise DPA for $L$ and ${\sim}$ that we introduce below is derived from the precise FWPM for $L$ and ${\sim}$. In general, for the construction of a DPA from an FWPM that captures the periodic part of $(L,{\sim})$, it is important that the individual priority mappings are compatible with each other.
Specifically, adding a prefix to a word should only lead to a smaller (i.e.~more significant) priority. This is expressed formally in the following definition.
\begin{definition}\label{def:monotonic}
	We call an FWPM $\family\gamma = \familyfor{\gamma}{{\sim}}$ \emph{monotonic} if $\gamma_u(vx) \leq \gamma_{uv}(x)$ for all $u,v \in \Sigma^*, x \in \Sigma^+$.
\end{definition}

\begin{lemma}\label{lem:precisecoloringmonotonic}
	Let $L$ be a regular $\omega$-language, and ${\sim}$ refine ${\sim}_L$. The precise FWPM $\family{\pi^{\sim}}$ of $(L,{\sim})$ is monotonic.
\end{lemma}
\begin{proof}
	Let $u,v \in \Sigma^*, x \in \Sigma^+$, and let $c = [u]_{\sim}$ and $c' = [uv]_{\sim}$. We use \cref{lem:normalizedDPAconnection}, so let $\AA = (Q,\Sigma,q_0,\delta,\kappa)$ be a normalized DPA for $L$. Denote by $Q_c$ and $Q_{c'}$ the sets of $c$-states and $c'$-states in $\AA$, respectively. By \cref{lem:normalizedDPAconnection}, $\pi_u^{\sim}(vx) = \max\{\Amin(q,vx) \mid q \in Q_c\}$. Further, for each $q \in Q_c$ we have $\Amin(q,vx) \le \Amin(\delta^*(q,v),x)$ because $\Amin(\delta^*(q,v),x)$ takes the minimal priority of a run that is the suffix of the run considered in $\Amin(q,vx)$.
	We conclude that
	\[
		\begin{array}{lcl}
			\pi_u^{\sim}(vx) & =   & \max\{\Amin(q,vx) \mid q \in Q_c\}            \\
			                 & \le & \max\{\Amin(\delta^*(q,v),x) \mid q \in Q_c\} \\
			                 & \le & \max\{\Amin(q',x) \mid q' \in Q_{c'}\}        \\
			                 & =   & \pi_{uv}^{\sim}(x)
		\end{array}
	\]
	where the second inequality follows from the fact that $\delta^*(q,v)$ is a $c'$-state for a $c$-state $q$.
\end{proof}

\subparagraph*{The Precise DPA.}
Our goal is now to construct, from a monotonic FWPM  $\family{\gamma} = \familyfor{\gamma}{{\sim}}$ that captures the periodic part of $(L,{\sim})$, a combined priority mapping $\bow\gamma$ that defines $L$.
We refer to $\bow\gamma$ as the \emph{join of $\family\gamma$}.

Recall that it is sufficient if $\bow\gamma$ works correctly on ultimately periodic words.
Since \(\family\gamma\) captures the periodic part of \(L\), for all \(uv^\omega\) such that \(uv \sim u\) we have that \(\gamma_c(v^\omega)\) is even if and only if \(v^\omega \in L_c\), where \(c\) is the \(\sim\)-class reached by \(u\).
The main problem when defining \(\bow\gamma\) for an input \(p \in \Sigma^+\) is that we only know the prefix \(p\) but not the ultimately periodic word \(uv^\omega\) that it is a prefix of.

Intuitively, the idea for computing the join for an input \(p\) is to consider all possible factorizations \(uv = p\) and their associated priorities \(\gamma_u(v)\).
Of those, the most significant one (i.e.~the smallest) then determines the priority of \(\bow\gamma(p)\).
Naively defining the join based on this, however, would just lead to the weak priority mapping for the initial class, since $\family{\gamma}$ is monotonic.
Instead, the join ensures that all these values $\gamma_u(v)$ are \emph{covered} by the priority sequence on $p$, where $\gamma_u(v)$ is covered if $\bow\gamma$ emits a priority less than or equal to $\gamma_u(v)$ on the $v$-part of the decomposition $p=uv$.
Obviously, with this definition, it is possible to always emit priority $0$. This clearly will cover all values $\gamma_u(v)$, but it will certainly not lead to a correct priority mapping for $L$.
Instead, the join always picks the least significant priority that is necessary to ensure that all values are covered.
Assume that the whole input is $w = uv^\omega$ with \(uv \sim u\) and $\gamma_u(v^\omega) = i$ where priority $i$ is already assumed after one iteration of $v$, so $\gamma_u(v) = i$. This means that $i = \gamma_u(v) = \gamma_{uv}(v) = \gamma_{uvv}(v) \cdots$. Then the join will infinitely often emit a priority $\le i$ in order to cover all these values.
Based on the monotonicity of $\family\gamma$ we can show in \cref{lem:join} that indeed $i$ will be the dominating priority that is emitted by $\bow\gamma$.

In the following example, we further try to illustrate the idea of the join, before subsequently formalizing it in \cref{def:join}.

\begin{example}\label{ex:join}
	Consider the FWPM \(\familyfor{\pi}{\sim}\) from \cref{ex:pisets} which is displayed on the left of \cref{fig:joinexample}: \(\pi_\varepsilon^\sim\) assigns \(0\) to a word \(v \in \Sigma^+\) precisely if \(v\) contains the infix \(aba\), \(1\) if \(v\) contains a \(b\) but no \(aba\) and \(2\) otherwise.
	Further, $\varepsilon$ is the only class.
	The table on the right of \cref{fig:joinexample} visualizes the computation of the join for all prefixes of the word \(p = abaaba\).
	The $i$-th letter of this word is denoted by $p_i$.
	As explained before, to compute $\bow\pi^\sim$ for some prefix $p_0\cdots p_i$ of $p$, we have to consider all values $\pi^\sim_{p_0\dotsc p_{j-1}}(p_j\dotsc p_i)$, which are given in the column for the corresponding value of $i$.
	Note that the subscript $p_0\dotsc p_j$, which denotes the class, does not play a role in this example, because there is only one class.
	We nevertheless spell it out in the example in order to match the formal definition.

	For finding $\bow\pi^\sim(p_0\cdots p_i)$, we have to ensure that all the values from the column for the corresponding $i$ are covered. The values that are not covered by a previous value of $\bow\pi^\sim$ are underlined.
	For $i=0$, no priority has been emitted by $\bow\pi^\sim$, so we just take the maximal value that covers $2$, which is $2$ itself.
	For $i=1$, the column contains two times priority $1$, and since up to now only $2$ has been emitted by $\bow\pi^\sim$, we obtain $\bow\pi^\sim(ab) = 1$.
	Similarly, because priority $0$ has to be covered in the next step, we obtain $\bow\pi^\sim(aba) = 0$.

	Now, for $i=3$, the values that need to be covered are the four values in the column for $i=3$. The value $0$ is coming from the decomposition $(\varepsilon,abaa)$ in which the periodic part is the whole word $abaa$. This value $0$ is already covered by $\bow\pi^\sim(aba) = 0$.
	Similarly, $\bow\pi^\sim(aba) = 0$ covers the values $\pi_{a}^\sim(baa) = 1$ and $\pi_{ab}^\sim(aa) = 2$.
	The only value from the column for $i=3$ that has not yet been covered, is the value
	$\pi_{aba}^\sim(a) = 2$, because $\bow\pi^\sim$ did not yet emit anything on this periodic part, which only includes the last letter.
	Hence, $\bow\pi^\sim(abaa) = 2$.

	Let us jump to $i=6$. For finding \(\bow\pi^\sim(abaaba)\), all values in the last column are considered. The first three $0$s are covered by $\bow\pi^\sim(aba) = 0$ because they correspond to decompositions whose periodic part contains position $i=2$.
	The other values have not yet been covered, and hence \(\bow\pi^\sim(abaaba) = 0\).
\end{example}

\begin{figure}
	\begin{center}
		\begin{minipage}{0.47\textwidth}
			\[
				\pi_\varepsilon^\sim(u) = \begin{cases}
					0 & \text{if \(u\) contains \(aba\)}               \\
					1 & \text{if \(u\) contains \(b\) but not \(aba\)} \\
					2 & \text{otherwise}
				\end{cases}
			\]
		\end{minipage}
		\hfill
		\begin{minipage}{0.5\textwidth}
			\begin{tabular}{rcccccc}
				$i$                                           & $0$                & $1$              & $2$              & $3$              & $4$              & $5$              \\
				$p_i$                                         & $a$                & $b$              & $a$              & $a$              & $b$              & $a$              \\
				\hline
				\(\pi^\sim_\varepsilon(p_0\dotsc p_i)\)       & \({\underline 2}\) & {$\underline 1$} & {$\underline 0$} & {$\green 0$}     & {$\green 0$}     & {$\green 0$}     \\
				{\(\pi^\sim_{p_0}(p_1\dotsc p_i)\)}           &                    & {$\underline 1$} & {$\green 1$}     & {$\green 1$}     & {$\green 1$}     & {$\green 0$}     \\
				{\(\pi^\sim_{p_0\dotsc p_1}(p_2\dotsc p_i)\)} &                    &                  & {$\underline 2$} & {$\green 2$}     & {$\green 1$}     & ${\green 0}$     \\
				{\(\pi^\sim_{p_0\dotsc p_2}(p_3\dotsc p_i)\)} &                    &                  &                  & ${\underline 2}$ & ${\underline 1}$ & ${\underline 0}$ \\
				{\(\pi^\sim_{p_0\dotsc p_3}(p_4\dotsc p_i)\)} &                    &                  &                  &                  & ${\underline 1}$ & {$\underline 0$} \\
				{\(\pi^\sim_{p_0\dotsc p_4}(p_5\dotsc p_i)\)} &                    &                  &                  &                  &                  & ${\underline 2}$ \\
				\hline
				\(\bow\pi^\sim(p_0\dotsc p_i)\)               & $2$                & $1$              & $0$              & $2$              & $1$              & $0$
			\end{tabular}
		\end{minipage}
	\end{center}
	\caption{
		On the left, the priority mapping \(\pi_\varepsilon^\sim\) associated with the FWPM \(\familyfor{\pi}{\sim}\) from \cref{ex:pisets} is shown.
		The table on the right displays the values of \(\pi^\sim_{p_0\dotsc p_{j-1}}(p_{j}\dotsc p_{i})\) where \(i\) is the length of a prefix of \(p\) and \(j\) is a possible first position inside the looping part of the guessed ultimately periodic word.\\
		In each column, we mark the value(s) that are not yet covered by underlining them.
	}
	\label{fig:joinexample}
\end{figure}

We are now ready to give the formal definition of the join.
Note that since it is defined inductively, we need to reference values of the join that have already been computed for prefixes of the input.
So while the definition speaks of partial priority mappings, the resulting mapping is guaranteed to be total.

\begin{definition} \label{def:join}
	Let $u \in \Sigma^+$ and $\gamma$ be a partial priority mapping that is defined for all non-empty prefixes of $u$.
	Let $y$ be a non-empty suffix of $u$, that is $u = xy$ for an $x \in \Sigma^*$.
	We say that $\gamma$ \emph{covers} $\family\gamma$ on the suffix $y$ of $u$ if $\gamma(xz) \leq \gamma_{x}(y)$ for some non-empty prefix $z$ of $y$.

	The priority mapping $\bow\gamma$ is defined inductively as follows.
	Assume that $\bow\gamma$ has been defined for all non-empty prefixes of $u \in \Sigma^*$, and let $a \in \Sigma$.
	Then $\bow\gamma(ua)$ is the maximal value such that $\bow\gamma$ covers $\family\gamma$ on all non-empty suffixes of $ua$.
\end{definition}

We first show that $\bow\gamma$ behaves correctly on ultimately periodic words in an auxiliary lemma.
From this we can easily deduce that $\bow\gamma$ defines the correct language.

\begin{lemma}\label{lem:join}
	Let $\family\gamma = \familyfor \gamma {\sim}$ be a monotonic FWPM, $u \in \Sigma^*$ and $v \in \Sigma^+$. If $u \mathrel{\sim} uv$, then $\bow\gamma(uv^\omega) = \gamma_{u}(v^\omega)$.
\end{lemma}
\begin{proof}
	First note that if $u'$ and $v'$ are such that $u'(v')^\omega = uv^\omega$ and $u' \mathrel{\sim} u'v'$, then $\gamma_{u}(v^\omega) = \gamma_{u'}((v')^\omega)$. This is implied by monotonicity of $\family\gamma$ (\cref{def:monotonic}): Assume w.l.o.g.\ that $u$ is a prefix of $u'$. Then by monotonicity, we have $\gamma_{u}(v^\omega) \le \gamma_{u'}((v')^\omega)$. Now choose $n$ such that $u'$ is a prefix of $uv^n$. Then again by monotonicity, we have $\gamma_{u'}((v')^\omega) \le \gamma_{uv^n}(v^\omega) = \gamma_{u}(v^\omega)$.

	This means that we can choose any representation of the given ultimately periodic word for proving the claim.

	We now turn to the claim of the lemma. Let $\bow\gamma(uv^\omega) = i$, then the least color assigned to infinitely many prefixes of $uv^\omega$ is $i$.
	Thus, we can rewrite $uv^\omega$ such that $\bow\gamma(u) = i$ and $\bow\gamma(ux) \geq i$ for all non-empty prefixes $x \sqsubset v^\omega$.
	Further, we assume that $u$ is picked to be the shortest prefix with this property.\\
	Consider any non-empty $x \sqsubseteq v^\omega$ with $\bow\gamma(ux) = i$.
	By definition of $\bow\gamma$, there is a suffix $z$ of $ux$ that would not be covered if we had $\bow\gamma(ux) > i$.
	Since $\bow\gamma(u) = i$, we conclude that $z$ is a suffix of $x$ (otherwise it would be covered by the value of $\bow\gamma(u)$).
	Thus, we can write $x = yz$ with $\gamma_{uy}(z) = i$.
	Since $\overline\gamma$ is monotonic, it follows that $\gamma_u(x) = \gamma_u(yz) \leq \gamma_{uy}(z) = i$.
	This inequality cannot be strict because otherwise the suffix $x$ of $ux$ would not be covered by $\bow\gamma$ since all $\bow\gamma$-values after $u$ are $\ge i$.
	We conclude that $\gamma_{u}(v^\omega) = i$ as desired.
\end{proof}

\begin{lemma}\label{lem:joinlanguage}
	Let $L$ be an $\omega$-regular language and ${\sim}$ refine ${\sim}_L$.
	If $\overline\gamma$ is a monotonic FWPM that captures the periodic part of $(L,{\sim})$, then $L(\bow\gamma) = L$.
	In particular, $L(\bow\pi^{\sim}) = L$.
\end{lemma}
\begin{proof}
	If follows from \cite{buchiOriginalPaper} that if $K \cap \mathbb{UP} = L \cap \mathbb{UP}$ for two regular $\omega$-languages $K$ and $L$, then $K = L$. Thus, it suffices to verify that for all $u \in \Sigma^*, v\in \Sigma^+$ holds $uv^\omega \in L$ if and only if $uv^\omega \in L(\bow\gamma)$. Consider the sequence of classes $[u], [uv], [uvv],\dotsc$. Because ${\sim}$ has only finitely many classes, there must exist $i < j \leq |{\sim}| + 1$ such that $[uv^i] = [uv^j]$. Thus, we can pick $x$ and $y$ such that $uv^\omega = xy^\omega$ and $xy\mathrel{\sim} x$.\\
	By definition, we have $xy^\omega \in L$ if and only if $y^\omega \in L_x$.
	Since $y \in \Stay x$ and $\overline\gamma$ captures the periodic part of $L$, it further holds that $y^\omega \in L_x$ iff $y^\omega \in L(\gamma_x)$.
	Applying \cref{lem:join} now gives us $\bow\gamma(xy^\omega) = \gamma_x(y^\omega)$, from which we immediately conclude that $xy^\omega \in L(\bow\gamma)$ if and only if $xy^\omega \in L$.

	The claim $L(\bow\pi^{\sim}) = L$ directly follows from \cref{lem:capperiodicpart} and \cref{lem:precisecoloringmonotonic}.
\end{proof}

We now explain how to construct a DPA that computes the priority mapping $\bow\gamma$ if $\family\gamma$ is a monotonic FWPM given by a family $\family\MM = \familyfor{\MM}{{\sim}}$ of Mealy machines.
This DPA, which we denote by $\bow\AA(\overline\MM)$, has one component for tracking ${\sim}$, and one component for each priority $i$ that checks if there is a suffix of the input read so far that has to be covered and produces priority $i$.
The minimal such $i$ is emitted as priority, and all components for priorities $\ge i$ are reset to start tracking suffixes from this point.
Note that the construction can be applied even if $\family\gamma$ is not monotonic (but then we do not have any guarantees on the behavior of the resulting DPA). We formalize this below.

Assume that the range of the mappings in $\family\gamma$ is $\{0,\ldots,k-1\}$.
First, we extract for each priority $i \in \{0,\ldots,k-1\}$ from each Mealy machine $\MM_c$ a DFA $\DD_{c,i} = (Q_{c,i}, \Sigma, \iota_{c,i}, \delta_{c,i}, F_{c,i})$ with $L(\DD_{c,i}) = \{u \in \Sigma^+ \mid \gamma_c(u) \leq i\}$.
This is achieved by taking the transition structure of $\MM_c$ and redirecting all transitions with output $\le i$ into an accepting sink state (recall that $\gamma_c$ is weak), and then minimizing the resulting DFA.
This step is illustrated in \cref{fig:FWPM-Mealy}, where the class index is omitted because there is only one class.
Note that the DFA $\DD_{c,k-1}$ always accepts every non-empty word, so it is omitted in the example. However, it is convenient to keep it in the formal construction.

To simplify notation, we assume that the $Q_{c,i}$ are pairwise disjoint, allowing us to write $\delta(q, a)$ and $q \in F$ to refer to $\delta_{c,i}(q, a)$ resp. $q \in F_{c,i}$ for the unique $c \in [{\sim}]$ and $i< k$ such that $q \in Q_{c,i}$.
We now define the DPA $\bow\AA(\overline\MM) = (\bow Q, \Sigma, \bow\iota, \bow\delta, \bow\kappa)$ with
\[
	\bow Q = Q_0 \times \cdots \times Q_{k-1} \times Q_{\sim},
	\ \text{where}\
	Q_i = \bigcup_{c \in [{\sim}]} (Q_{c,i} \setminus F_{c,i})
	\ \text{and}\
	\bow\iota = (\iota_{[\varepsilon], 0}, \dotsc, \iota_{[\varepsilon],k-1},[\varepsilon]).
\]
In this product, we refer to the $i$-th component as the \emph{component for priority $i$} for $i \in \{0, \ldots, k-1\}$.
For each transition from some state $\tilde{q} = (q_0, \dotsc, q_{k-1},c) \in Q$ on a symbol $a \in \Sigma$, we begin by computing the reset point $r < k$, which corresponds to the least index of a DFA that reaches a final state, meaning formally $r = \min\{i < k \mid \delta(q_i, a) \in F\}$.
Note that $r$ is always defined because the DFA for priority $k-1$ accepts everything.

The reset point directly determines the priority of the transition, meaning $\bow\kappa(\tilde{q}, a) = r$.\\
For computing the target of the transition, we first advance each $q_j$ for $j < r$ by $a$ in the respective DFA and then reset all other states to the initial state of the appropriate DFA, i.e.

\[
	\bow\delta((q_0, \dotsc, q_{k-1},c), a) = (\delta(q_0,a), \dotsc, \delta(q_{r-1},a), \iota_{c', r}, \dotsc, \iota_{c', k-1},c') \quad\text{where } c' = [ca]_{\sim}.
\]

In \cref{fig:FWPM-Mealy} the resulting DPA for \cref{ex:pisets} is shown on the right-hand side.
We omit the component for the ${\sim}$-class and the component for priority $2$, because they both only consists of one state.
The precise DPA for \cref{ex:twoclasses} is shown in \cref{fig:twoclassexample}.

\begin{figure}
	\begin{center}
	\includegraphics[width=\textwidth]{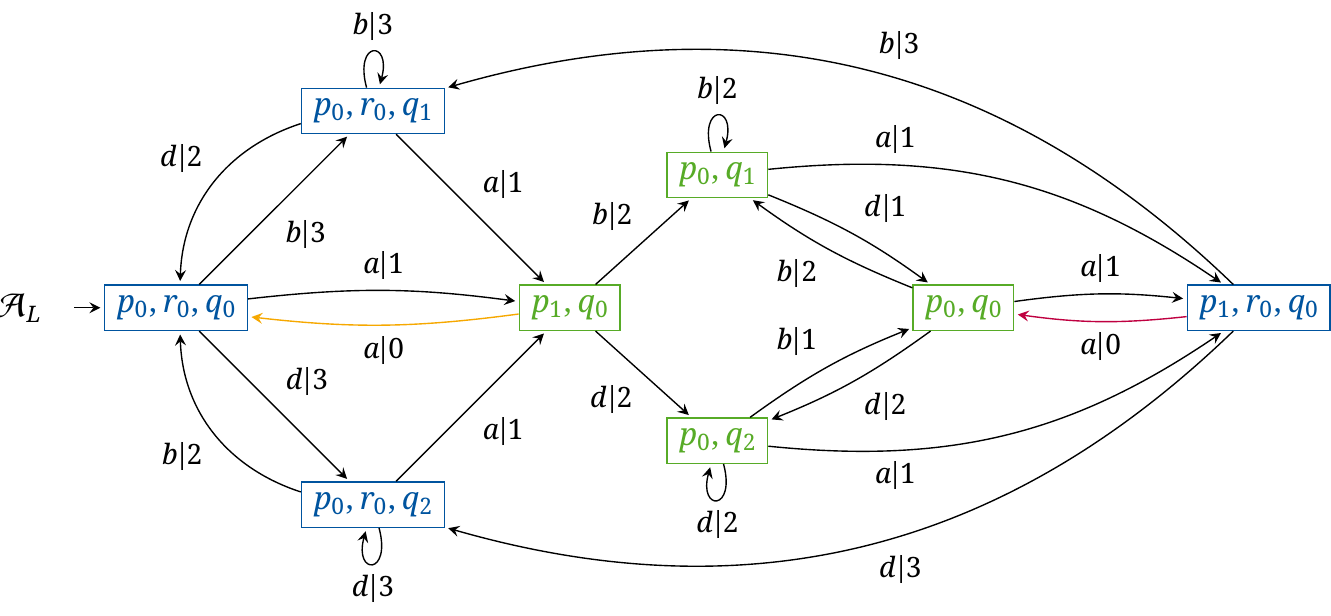}
	\end{center}
	\caption{
	The precise DPA for the language $L$ from example \cref{ex:twoclasses}.
	The DPA is built using the right congruence ${\sim}_L$ and DFAs for the sets $P_{c,i}^{{\sim}_L}(L)$ for $i \leq 3$ and $d \in [{\sim}_L]$, which are depicted in \cref{fig:twoclassexampleparts}.
	The states of $\AA_L$ are tuples, where the $i$-th position from the front is a state belonging to the component for priority $i$, i.e.~a DFA which keeps track of level $i$ of the parity decomposition.
	We omit a component if the corresponding DFA accepts $\Sigma^+$ and also do not include the last component, instead depicting states belonging to class $[\varepsilon]_{{\sim}_L}$ in \blue{blue} and those belonging to $[a]_{{\sim}_L}$ in {\color{green}green}.
	Note that reading $a$ from the initial state, the first component is not reset, as the corresponding DFA still waits on completing an infix $aa$.
	The DFA $\CC$ from \cref{fig:twoclassexampleparts} tracks priority $2$ from $[\varepsilon]_{{\sim}_L}$ and priority $1$ from $[a]_{{\sim}_L}$.
	Hence its states appear in two different positions in the tuples.
	The DFA $\BB$ tracks priority $0$ for both classes, so the component for priority $0$ is reset to $p_0$ for both classes on transitions with priority $0$ (the {\color{orange} orange} and {\color{purple}purple} $a$-transitions).
	}
	\label{fig:twoclassexample}
\end{figure}

The following lemma establishes that $\bow\AA(\MM)$ indeed computes the join.

\begin{lemma}\label{lem:bowacomputesjoin}
	If $\overline\MM$ is a family of Mealy machines representing a monotonic FWPM $\overline\gamma$, then $\bow\AA(\overline\MM)(u) = \bow\gamma(u)$ for each $u \in \Sigma^+$. In particular, $L(\bow\AA(\overline\MM)) = L(\bow\gamma)$.
\end{lemma}

\begin{proof}
	We write $\bow\AA$ for $\bow\AA(\overline\MM)$.
	Further, for avoiding case distinctions, we let $\bow\gamma(\varepsilon) := \bow\AA(\varepsilon) := 0$. Note that this has no influence on the claim of the lemma because the claim is only on non-empty words.

	We first prove the following characterization of $\bow\gamma$, and then show that this is precisely what $\bow\AA(\overline\MM)$ computes:

	\underline{Claim:} For $u \in \Sigma^+$, we have that $\bow\gamma(u)$ is the minimal $i$ such that there are $x \in \Sigma^*$, $y \in \Sigma^+$ with $u = xy$ and
	\begin{enumerate}[(1)]
		\item $\gamma_x(y) = i$,
		\item $\forall z \in \Sigma^+: z \sqsubsetneq y \Rightarrow \bow\gamma(xz) > i$, and
		\item $\bow\gamma(x) \le i$.
	\end{enumerate}
	(Where $z \sqsubsetneq y$ means that $z$ is a prefix of $y$ but not equal to $y$.)

	From the definition of $\bow\gamma$ it follows that $\bow\gamma(u)$ is the minimal $i$ for which $u=xy$ with properties (1) and (2) exist, because these are precisely the suffixes of $u$ that need to be covered by the value $\bow\gamma(u)$.
	So let $i$ be minimal such that there are $x \in \Sigma^*$, $y \in \Sigma^+$ with $u = xy$, $\gamma_x(y) = i$, and $\bow\gamma(xz) > i$ for all strict non-empty prefixes of $y$.
	Furthermore, choose $x$ and $y$ such that $y$ has maximal length with these properties. We now show that (3) is also satisfied, thus proving the claim.

	Assume for contradiction that $\bow\gamma(x) = j > i$. Then $x$ is non-empty because we defined $\bow\gamma(\varepsilon) = 0$. Let $x = x'a$ with $a \in \Sigma$. Since $\family\gamma$ is monotonic, $i' := \gamma_{x'}(ay) \le \gamma_x(y) = i$. Now let $z$ be a nonempty strict prefix of $ay$. If $z = a$, then $\bow\gamma(x'z) = \bow\gamma(x) = j > i$. If $z = az'$ with nonempty $z'$, then $\bow\gamma(x'z) = \bow\gamma(xz') > i$ because $x,y$ and $i$ satisfy property (2). Overall, we obtain that $x',ay$ and $i'$ satisfy properties (1) and (2). If $i' < i$, this contradicts the choice of $i$ as the minimal such value, and if $i' = i$, this contradicts the choice of $y$ as longest suffix of $u$ such that $x,y$ and $i$ have properties (1) and (2).

	\medskip

	Now let us show that the priority of $\bow\AA$ after reading $u$ is the one characterized in the claim.
	We assume inductively that $\bow\AA$ has produced the same priorities as $\bow\gamma$ on all strict non-empty prefixes of $u$.
	Let $\bow\AA(u) =: i$. Then the component for priority $i$ in $\bow\AA$ has reached a final state of the respective DFA that is currently running in that component, and none of the DFAs in the components of smaller priorities has reached a final state. Let $x$ be the longest strict prefix of $u$ such that $\bow\AA(x) \le i$ (by our convention $\bow\AA(\varepsilon) = 0$, such an $x$ always exists), and let $y \in \Sigma^+$ be such that $u=xy$.

	By definition of $\bow\AA$, the component for priority $i$ is reset to the initial state of $\DD_{x,i}$ after reading $x$, and since $x$ is the longest prefix on which priority $\le i$ is produced, none of the components for priorities $\le i$ is reset after $x$. We had seen above that the component for priority $i$ in $\bow\AA$ has reached a final state, which means that $\gamma_x(y) \le i$ by definition of $\DD_{x,i}$. Since the components for priorities smaller than $i$ have not reached a final state, we obtain by monotonicity that $\gamma_x(y) = i$. Hence, $x,y$ and $i$ satisfy the properties (1)--(3) of the claim, and thus $\bow\AA(u) = i \ge \bow\gamma(u)$.

	Now assume that there is $j < i$ that satisfies (1)--(3) for some appropriate $x$ and $y$. Then, along the same lines as above, the component for priority $j$ is reset after that $x$, and reaches a final state after reading $y$. This means that $\bow\AA$ would output a priority $\le j$.
\end{proof}

We know from \cref{mealysize} that the precise FWPM $\overline{\pi^{\sim}}$ for $(L,{\sim})$ can be computed by a family of Mealy machines. This enables us now to define the central object of this section.
\begin{definition} \label{def:preciseDPA}
	Let $L$ be a regular $\omega$-language, and let ${\sim}$ be a right congruence that refines ${\sim}_L$.
	The \emph{precise DPA $\AA_{L,{\sim}}$ for $(L, {\sim})$} is the minimal DPA whose priority mapping is the join $\bow\pi^{\sim}$ of the precise FWPM $\overline{\pi^{\sim}}$ for $(L,{\sim})$. The \emph{precise DPA $\AA_L$ for $L$} is $\AA_{L,{\sim}_L}$.
\end{definition}
The DPA on the right-hand side of \cref{fig:FWPM-Mealy} is minimal as a Mealy machine, so it is the precise DPA for the language from \cref{ex:pisets}. Similarly, the DPA in \cref{fig:twoclassexample} is the precise DPA for the language from \cref{ex:twoclasses}.

Note that $L(\AA_{L,{\sim}}) = L$ because the precise FWPM $\family{\pi^{\sim}}$ for $L$ and ${\sim}$ captures the periodic part of $(L,{\sim})$ according to \cref{lem:capperiodicpart}, and hence $L(\bow\pi^{\sim}) = L$ by \cref{lem:joinlanguage}. Since, by definition, the priority mapping of $\AA_{L,{\sim}}$ is $\bow\pi^{\sim}$, we obtain that $L(\AA_{L,{\sim}}) = L$.
Furthermore, given an arbitrary DPA $\AA$ for $L$, and ${\sim}$ given as transition system, one can construct $\AA_{L,{\sim}}$.
For use in the \cref{dpalearner}, we give an upper bound on the size of the precise DPA that is constructed from another DPA.
\begin{theorem}\label{boundpreciseDPA}
	Let $\AA$ be a DPA with priorities in $\{0, \ldots, k-1\}$ accepting some language $L$, and ${\sim}$ be a right congruence with ${\sim}_\AA \refine {\sim} \refine {\sim}_L$. If $\AA$ has at most $d$ many $c$-states for each class $c$ of ${\sim}_L$, then $|\AA_{L,{\sim}}| \le m(2^d-1)^{k-1}$, where $m$ is the number of ${\sim}$ classes.
\end{theorem}
\begin{proof}
	We can assume that $\AA = (Q,\Sigma,\iota,\delta,\kappa)$ is normalized (this does not change the transition structure and does not increase the number of priorities). The precise DPA $\AA_{L,{\sim}}$ can be constructed from the family $\family\MM$ of Mealy machines for $\overline{\pi^{\sim}}$ by building $\bow\AA(\family\MM)$.
	In the construction of $\bow\AA(\family\MM)$, the first step is to extract the DFAs $\DD_{c,i}$ for each class $c$ and each priority $i$ from the Mealy machines in $\family\MM$. The DFA $\DD_{c,i}$ accepts a word $w$ if $\pi_c^{{\sim}}(w) \le i$. From \cref{lem:normalizedDPAconnection} we obtain that $\pi_c^{{\sim}}(w) \le i$ if from every $c$-state $q$ in $\AA$, the run on $u$ has visited a priority $\le i$.
	We can directly construct $\DD_{c,i}$ from $\AA$ without going through the Mealy machines as follows. The states of $\DD_{c,i}$ are sets $P \subseteq Q$ such that all states in $P$ are pairwise ${\sim}$-equivalent. The initial state is the set $Q_c$ of $c$-states in $\AA$. For a state $P$ of $\DD_{c,i}$, the $a$-successor of $P$ is $\{\delta(p,a) \mid p \in P \text{ and } \kappa(p,a) > i\}$. Since ${\sim}_\AA$ refines ${\sim}$, the property that all states in $P$ are ${\sim}$-equivalent is preserved by the transitions. The only accepting state of $\DD_{c,i}$ is $\emptyset$.

	The states of $\bow\AA(\family\MM)$ are tuples $(P_0, \ldots, P_{k-1},c)$ where each $P_i$ is a non-accepting state of some $\DD_{c',i}$, and $c$ is the ${\sim}$-class of the input read so far. By construction, all $\AA$-states in $P_i$ are in $Q_c$. Further, in $\DD_{c,{k-1}}$, every transition from the initial state directly leads to the accepting state. So for each $c$ there are at most $(2^d-1)^{k-1}$ many states in $\bow\AA(\family\MM)$, which shows the claimed bound.
\end{proof}

We now present a family of examples showing that the precise DPA can be exponential in the size of a minimal DPA. This is already witnessed by Büchi automata, so DPAs with priorities $\{0,1\}$. The example is taken from \cite[Proposition~14]{BohnL22}.

For a number $d \ge 1$, consider the alphabet $\Sigma_d = \{a_0,\dotsc,a_{d-1}\}$ and let $L_d \subseteq \Sigma_d^\omega$ be the language consisting of all $\omega$-words that contain each symbol from $\Sigma_d$ infinitely often. One easily verifies that ${\sim}_{L_d}$ has only one class.

In a first approach for constructing a DPA for $L_d$, one would build an automaton that tracks precisely which symbols it has seen. This means that the states are subsets of $\Sigma_d$, with $\emptyset$ as initial state, and transitions adding the processed symbol to the current state. Whenever a transition would result in a state for the full set $\Sigma_d$, the DPA emits priority $0$, and resets the state to $\emptyset$. All other transitions have priority $1$.

Indeed, this is what the precise DPA does: The precise FWPM for $L_d$ and ${\sim}_{L_d}$ (consisting only of one weak priority mapping since there is only one class), maps a finite word to priority $0$ if it contains all symbols from $\Sigma_d$, and to $1$ otherwise. The precise DPA resulting from that precise FWPM then behaves exactly as described above. So it has $2^d-1$ states.

However, one can build a DPA with only $d$ states that waits for the next symbol from $\Sigma_d$ for some fixed order on $\Sigma_d$, and resets to the initial state with priority $0$ if it has reached the last symbol in this order.
For the standard order on $\{0,\ldots,d-1\}$ and state set $\{q_0, \ldots, q_{d-1}\}$, the resulting transition function is $\delta(q_h,a_h) = a_{h+1 \mod d}$ with priority $0$ if $h=d-1$ and $1$ otherwise, and  $\delta(q_h,a_{h'}) = q_h$ if $h \not= h'$.

This family of examples might also convey some further intuition why we chose the name ``precise DPA'': It emits the accepting priority $0$ precisely when all symbols have occurred. The small DPA with only $d$ states emits priority $0$ if all symbols have occurred in a specific order, so it would not emit priority $0$ on the finite word $a_2a_1a_0$ (for $d=3$) but on $a_0a_1a_2$. It seems hard to define such a behavior just from the language definition in a canonical way.

The construction of the precise DPA for $L$ sets out a rough road map for the construction of a DPA from a collection $S$ of positive and negative example words: Extract a right congruence ${\sim}$ from $S$ as candidate for ${\sim}_L$ (this is known how to do it). Then extract Mealy machines $\family\MM$ that capture the periodic part of $S$ and ${\sim}$. Construct $\bow\AA(\family\MM)$.

In \cref{forcs} we present some results that enable us to realize the extraction of Mealy machines. The problem with the last step is that the size of $\bow\AA(\family\MM)$ is of order $|\family\MM|^k$, where $k$ is the number of priorities. In a polynomial time procedure we therefore cannot simply construct $\bow\AA(\family\MM)$ from the Mealy machines because this would result in an exponential step.
To overcome this issue, we show in the following that the size of the representation of the color sequence produced by $\bow\AA(\family\MM)$ on an ultimately periodic word $uv^\omega$ is polynomial in $|uv|$ and $|\family\MM|$. This allows us to construct a DPA of polynomial size, which approximates $\bow\AA(\family\MM)$ in the sense that it emits the correct color sequence for a set of ultimately periodic words (the details how it is used are given in \cref{dpalearner}). In the following, we first state an auxiliary lemma.

\begin{lemma}\label{lem:shortacceptedprefix}
	Let $v \in \Sigma^+$.  A DFA $\DD$ with $n$ states accepts a prefix $v^\omega$ from a state $q$ if, and only if, there exists a non-empty $x \in \prf(v^\omega) \cap L(\DD)$ with $|x| \leq n\cdot|v|$.
\end{lemma}
\begin{proof}
	Consider the run $q = q_0 \xrightarrow{v} q_1 \xrightarrow{v} q_2 \xrightarrow{v} \cdots$ of $\DD$ on $v^\omega$. As $\DD$ has $n$ states, there must be positions $i < j \leq n$ at which a state repeats, i.e.~$q_i = q_j$. Either a final state is visited before position $j$, or the DFA loops on $v^\omega$ without visiting a final state at all.
\end{proof}

Note that the FWPM in the following lemma is not required to be monotonic.
\begin{lemma}\label{lem:polynomialjoincoloring}
	Let $u \in \Sigma^*, v \in \Sigma^+$ and $\overline\MM$ be a family of Mealy machines computing an FWPM $\overline\gamma$.
	The sequence of priorities produced by $\bow\AA(\overline\MM)$ on $uv^\omega$ can be written as an ultimately periodic word $rs^\omega$ where $|rs|$ is polynomial in $|uv|$ and $|\overline\MM|$. Furthermore, $r$ and $s$ can be computed in polynomial time.
\end{lemma}
\begin{proof}
	To simplify notation, we let $\AA := \bow\AA(\overline\MM)$ and use $\gamma$ to refer to the priority mapping computed by $\AA$.  Further, we assume that $uv \mathrel{\sim} u$. Otherwise, we can rewrite the representation of $uv^\omega$ accordingly with only a polynomial blow-up of the representation.

	We decompose $uv^\omega$ into a sequence of non-empty finite factors $x_0,x_1 \cdots$ by making the next cut for each $x_h$ right after the first time $\AA$ emits the smallest priority $i_h$ it will emit on the remaining suffix $w_h$. Formally, we define the sequence inductively with $y_0 := \varepsilon$, $w_0 := uv^\omega$, and if $y_h,w_h$ with $y_hw_h = uv^\omega$ are given, we let $x_h$ be the shortest prefix of $w_h$ such that $\gamma(y_hx_h) \le \gamma(y_hx)$ for all non-empty prefixes $x$ of $w_h$.
	Then let $i_h := \gamma(y_hx_h)$,  $y_{h+1} := y_hx_h$ and $w_{h+1}$ such that $w_h = x_hw_{h+1}$.

	We first explain why $i_h,x_h$ can be computed in polynomial time, and why the length of each $x_h$ is polynomial in $|uv|$ and $|\overline\MM|$.

	By definition, the states of $\AA$ are tuples of the form $(p_{0}, \dotsc, p_{k-1},c)$, where $c$ is a class of ${\sim}$ and the $p_{i}$ are states of DFAs whose size is bounded by the size of the Mealy machines in $\family\MM$. We already know that $\AA$ only emits priorities $\ge i_{h-1}$ on the remaining suffix $w_h$ (letting $i_{-1}= 0$ for $h=0$). And after the prefix $y_h$, all components for priorities $j \ge i_{h-1}$ are set to the initial state of the DFAs for the current ${\sim}$-class $c_h$. To determine $i_h$, we can check in increasing order starting from $j = i_{h-1}$, whether $\DD_{c_h,j}$ accepts a prefix of $w_h$. By  \cref{lem:shortacceptedprefix} this can be done in polynomial time, and if there is such a prefix, its length is bounded by $|u|+ |v||\DD_{c_h,j}|$.

	We now explain why we only need to compute a polynomial number of the factors.

	Since each $i_h$ is the smallest priority that appears on the remaining suffix, we have $i_0 \le i_1 \le \cdots$. So for $h \ge k-1$, all $i_h$ are the same, say $i$.
	Since $\AA$ resets all components for priorities $\ge i$ after emitting $i$, we get for all $h,h' \ge k$ with $w_h = w_{h'}$ that $x_h = x_{h'}$. Furthermore, the sequence of priorities emitted on the factors $x_h$ and $x_{h'}$ is also the same because it only depends on the states in the components for priorities $\ge i$, which are reset each time that $i$ is emitted.

	Since the $w_h$ are suffixes of $uv^\omega$, such a repetition of a suffix happens after at most $|u|+|v|$ steps. But if $w_h = w_{h'}$ and $x_h = x_{h'}$, then also $w_{h+1} = w_{h'+1}$ and thus the sequence of priorities becomes periodic at this point.

	Thus, in order to compute an ultimately periodic representation of the priority sequence of $\AA$ on $uv^\omega$, it suffices to compute the sequence of the $w_h,y_h,x_h,i_h$ up to a point $h_1$ such that there is $h_0 < h_1$ with $w_{h_1} = w_{h_0}$ and $i_{h_1} = i_{h_0}$. The resulting prefix $x_0x_1 \cdots x_{h_1-1}$ is of polynomial length in $|uv|$ and $|\family\MM|$ according to the above explanations. Then we let $r$ be the priority sequence of $\AA$ on $x_0\cdots x_{h_0-1}$ and $s$ the priority sequence of $\AA$ on $x_{h_0}\cdots x_{h_1-1}$. These can be computed on the fly without fully constructing $\AA$, so their computation is polynomial in $|uv|$ and $|\family\MM|$.
\end{proof}

\section{Families of Right Congruences}\label{forcs}
In Section~\ref{preciseDPA} we have seen that we can construct a DPA for $L$ from ${\sim}_L$ or any right congruence ${\sim}$ that refines ${\sim}_L$, and from Mealy machines for the colorings $\family\pi^{\sim} = (\pi^{\sim}_c)_{c \in [{\sim}]}$. Our goal is to learn such Mealy machines from given examples. For this purpose, we explain in this section how to define the Mealy machines using the formalism of families of right congruences (FORCs) \cite[Definition~5]{syntacticcongruence}. We start by giving some intuition, why and how we use this formalism, and then go into the formal details.

The usual way to obtain a method that can construct a class of transition systems in the limit from examples, is to characterize them by a right congruence on finite words, such that non-equivalence of two words is witnessed by a finite set of positive and negative examples for $L$. For example, the non-equivalence of $x,y$ for ${\sim}_L$ is witnessed by two examples $xuv^\omega \in L \Leftrightarrow yuv^\omega \not\in L$. Then an appropriate method can extract the desired transition system if enough non-equivalences are specified by the examples.
We did not manage to characterize the transition structure of a minimal Mealy machine for a coloring $\pi^{\sim}_c$ in such a way. To illustrate this, consider the case that ${\sim} = {\sim}_L$ has only one class $c$, and that only the priorities $0$ and $1$ are used. Write $P_0$ for $P_{c,0}^{\sim}$ and similarly for $P_1$. A minimal Mealy machine for $\pi^{\sim}_c$ consists of an initial part that assigns $1$ to the words in $P_1(L)$, and a sink state that that is reached for all words in $P_0(L)$. So two words $x,y$ lead to different states in this Mealy machine if and only if there is an extension $z$ with ($xz \in P_0(L) \Leftrightarrow yz \not\in P_0(L)$). Substituting the definition of $P_0(L)$, we obtain that $x$ and $y$ lead to a different state if and only if
\(
\exists z \in \Sigma^*:\, (\forall z' \in \Sigma^*:\, (xzz')^\omega \in L) \Leftrightarrow  (\exists z' \in \Sigma^*:\, (yzz')^\omega \not\in L).
\)

This characterization of non-equivalence has a universal quantifier, and thus cannot be witnessed by  a finite set of examples.
However, we can use a condition that is implied by the above one, which means that we characterize a possibly larger transition system. If $x$ and $y$ satisfy the above non-equivalence condition, then, in particular, there are $z,z' \in \Sigma^*$ with $(xzz')^\omega \in L \Leftrightarrow  (yzz')^\omega \not\in L$. Rewriting this with a single word $z$ instead of $zz'$, we get the condition
\(
\exists z \in \Sigma^*:\, (xz)^\omega \in L \Leftrightarrow  (yz)^\omega \not\in L
<\)
for non-equivalence of $x$ and $y$, which is implied by the first condition.
Generalizing this to a right congruence ${\sim}$ with several classes, directly leads to the notion of a canonical FORC for ${\sim}$, as defined below.

A \emph{family of right congruences (FORC)} \cite{syntacticcongruence} is a tuple $\FF = ({\sim}, (\approx_c)_{c \in \classes {\sim}})$ such that ${\sim}$ is a right congruence over $\Sigma^*$, each $\approx_c$ for $c \in \classes {\sim}$ is a right congruence over $\Sigma^+$, and for all $x,y \in \Sigma^*$ and $c \in \classes {\sim}$ it holds that $x \approx_c y$ implies $ux \mathrel{\sim} uy$, for an arbitrary element $u \in c$. Instead of writing $\approx_c$ for a class $c$, we sometimes also write $\approx_u$ for an arbitrary word $u \in c$. We extend $\approx_c$ from $\Sigma^+$ to $\Sigma^*$ by adding a new class that contains only $\varepsilon$.

Adopting the terminology of \cite{Klarlund94,AngluinF16}, we call ${\sim}$ the \emph{leading right congruence} (LRC) of $\FF$, and each $\approx_c$ the \emph{progress right congruence} (PRC) for $c$. FORCs can be used as acceptors for $\omega$-languages (see \cite[Definition~6]{syntacticcongruence}, \cite[Definition of LFORC]{Klarlund94}, \cite[Definition~6]{AngluinF16})
but since this is not relevant for our results, we do not go into the details of it.
We are interested in the \emph{canonical FORC} for $L \subseteq \Sigma^\omega$ and an LRC ${\sim}$ that refines ${\sim}_L$. For $u \in \Sigma^*$, we define the \emph{canonical PRC} ${\simeq}_u$ for $L$ and ${\sim}$ by
\[
	x {\simeq}_u y \text{ iff } ux \mathrel{\sim} uy \text{ and } \forall z \in \Sigma^*:\, uxz \mathrel{\sim} u \Rightarrow (u(xz)^\omega \in L \Leftrightarrow u(yz)^\omega \in L).
\]
The \emph{syntactic FORC} of $L$ \cite[Definition~9]{syntacticcongruence} is the canonical FORC for $L$ and ${\sim}_L$. As a running example, consider the language $L$ from \cref{ex:pisets}. Then ${\sim}_L$ has only one class, and the PRC for this class is shown in \cref{exampleforc}. The transition structure of the Mealy machine in \cref{fig:FWPM-Mealy} is obtained by merging some of the classes in the PRC.
\begin{figure}
	\centering
	\includegraphics[width=0.7\textwidth]{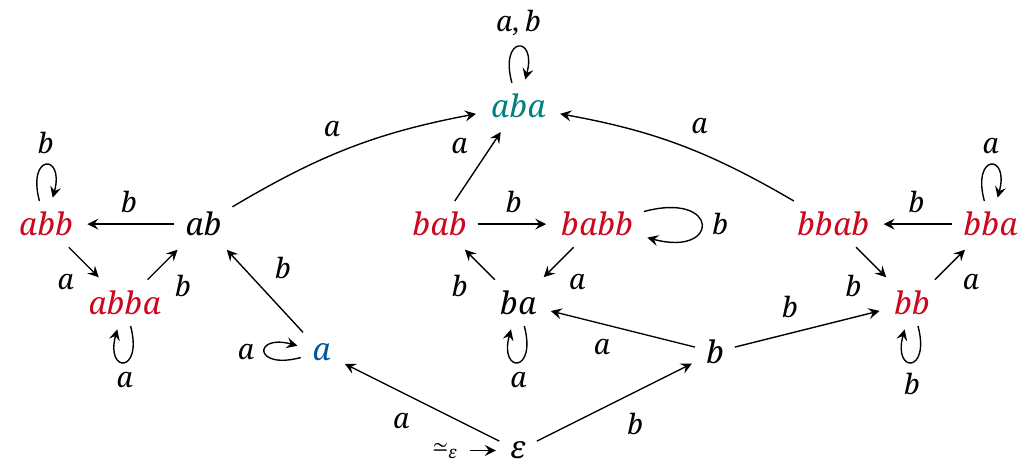}
	\caption{The PRC ${\simeq}_\varepsilon$ of the syntactic FORC for the language from \cref{ex:pisets}. The colors code the idempotent classes that are positive ($aba$ and $a$) and negative.}
	\label{exampleforc}
\end{figure}

The syntactic FORC is a canonical object for representing regular $\omega$-languages: In \cite[Theorem~22 and Theorem~24]{syntacticcongruence} it is shown that $L \subseteq \Sigma^\omega$ is regular if and only if its syntactic FORC is finite, and the syntactic FORC is the coarsest FORC with ${\sim}_L$ as leading right congruence that recognizes $L$.

We define the size of a FORC as the sum of the sizes of the LRC and the PRC.
Note that if the syntactic FORC of $L$ is finite, then also the canonical PRCs for each ${\sim}$ that refines $L$ are finite, as a direct consequence of the definition of the canonical FORC for ${\sim}$.

We now explain how we can use the PRCs ${\simeq}_c$ of the canonical FORC for $L$ and ${\sim}$ in order to obtain a Mealy machine for the priority mappings $\pi_c^{\sim}$. So in the following, let $L \subseteq \Sigma^\omega$ be regular with parity complexity $k$, let ${\sim}$ be a refinement of ${\sim}_L$, and let ${\simeq}_c$ be the canonical PRC for each class $c$ of ${\sim}$.
We call a class $[x]_{{\simeq}_c}$ with $x \in E_c$ \emph{positive} if $x^\omega \in L_c$, and \emph{negative} otherwise (by the definition of ${\simeq}_c$ this is independent of the chosen word from the class). In the example from \cref{exampleforc}, the positive classes are the ones of $a$, $ab$, $ba$, and $aba$.

\phantomsection\label{kappacoloring}
We define a coloring $\kappa_c$ on the classes of ${\simeq}_c$ (so on the states of the transition system defined by ${\simeq}_c$) such that the Mealy machine that outputs the color of its target state on each transition defines $\pi_c^{\sim}$ (see \cref{PRCcoloring}).
For this coloring, idempotent classes of ${\simeq}_c$ play a central role. We call a word $x \in \Sigma^+$ \emph{idempotent} in ${\simeq}_c$ if $x \in E_c$ and $x \mathrel{{\simeq}_c} xx$.
The following lemma states two useful properties on idempotent words in ${\simeq}_c$.
\begin{lemma}\label{lem:idempotent}
	If $x \in \Sigma^+$ is idempotent in ${\simeq}_c$, then each $y \in \Sigma^+$ with $y \mathrel{{\simeq}_c} x$ is idempotent in ${\simeq}_c$. Furthermore, for each $x \in E_c$, there is a number $n \ge 1$ such that $x^n$ is idempotent in ${\simeq}_c$.
\end{lemma}
\begin{proof}
	We begin with the first claim.
	Let $u \in c$. First note that $y \in E_c$ because $x \in E_c$ and $y \mathrel{{\simeq}_c} x$ implies $ux \sim uy$. For showing $y \mathrel{{\simeq}_c} yy$, consider $z \in \Sigma^+$ with $uyz \sim u$. Then
	\begin{alignat*}{5}
		 & u(yz)^\omega \in L\quad &  & \Leftrightarrow \quad &  & u(xz)^\omega   &  & \in L  \quad\qquad &  & \text{because $y \mathrel{{\simeq}_c} x$}                        \\
		 &                         &  & \Leftrightarrow       &  & u(xxz)^\omega  &  & \in L              &  & \text{as $x \in E_c$ is idempotent}                              \\
		 &                         &  & \Leftrightarrow       &  & u(yxz)^\omega  &  & \in L              &  & \text{$y \mathrel{{\simeq}_c} x$}                                \\
		 &                         &  & \Leftrightarrow       &  & uy(xzy)^\omega &  & \in L              &  & \text{$uy \mathrel{\sim} u$ and ${\sim}$ refines ${\sim}_L$}     \\
		 &                         &  & \Leftrightarrow       &  & uy(yzy)^\omega &  & \in L              &  & \text{$y \mathrel{{\simeq}_c} x$}                                \\
		 &                         &  & \Leftrightarrow       &  & u(yyz)^\omega  &  & \in L              &  & \text{as $uy \mathrel{\sim} u$ and ${\sim}$ refines ${\sim}_L$}.
	\end{alignat*}

	We now show the second claim.
	Since ${\simeq}_c$ has finitely many classes, there are $n_1,n_2 \ge 1$ such that $x^{n_1} \mathrel{{\simeq}_c} x^{n_1+n_2}$. Then it is not hard to see that for $n:=n_1n_2$ we have $x^n \mathrel{\simeq} x^{2n}$.
\end{proof}
Accordingly, we call a class of ${\simeq}_c$ idempotent if one (and hence all) of its members are idempotent. In the example PRC in \cref{exampleforc}, the only non-idempotent classes are the classes of $b$, $ab$ and $ba$.

The algorithm for computing $\kappa_c$ considers positive and negative idempotent classes (the classification of non-idempotent classes as positive or negative does not play any role). In the example from \cref{exampleforc}, the positive idempotent classes are the ones of $a$ and $aba$. The other two positive classes $ab$ and $ba$ are not idempotent.

Now let $\kappa_c: [{\simeq}_c] \rightarrow \{0,\ldots,k\}$ be defined by
\begin{itemize}
	\item $\kappa_c([x]_{{\simeq}_c}) = 0$ if there is no $z$ such that $[xz]_{{\simeq}_c}$ is a negative idempotent class. In other words, if no negative idempotent class is reachable from $x$.
	\item Let $i \ge 1$ and assume that $\kappa_c$ has already been defined for all values $\{0,\ldots,i-1\}$. Then $\kappa_c([x]_{{\simeq}_c}) = i$ for $i \ge 1$ if each idempotent class of the form $[xz]_{{\simeq}_c}$
	      \begin{itemize}
		      \item either already has a $\kappa_c$ value less than $i$,
		      \item or is positive iff $i$ is even.
	      \end{itemize}
\end{itemize}
In the example from \cref{exampleforc}, the class of $aba$ is assigned $0$. Then all remaining classes of words that contain $b$ are assigned value $1$, and finally, the class of $a$ is assigned $2$.

The following lemma shows that this indeed gives the desired coloring of $\Sigma^+$. In particular, all classes of ${\simeq}_c$ are assigned a value by the above procedure.
\begin{lemma}\label{PRCcoloring}
	For all $x \in \Sigma^+$: $\kappa_c(x) = i$ if and only if $\pi_c^{\sim}(x) = i$.
\end{lemma}
\begin{proof}
	The proof goes by induction on $i$. For the base case, note that
	$x \in P_{c,0}^{\sim}(L)$ if and only if $(xz)^\omega \in L_c$ for all $z \in \Sigma^*$ with $xz \in E_c$. In particular, this means that all idempotent classes of ${\simeq}_c$ that are reachable from the class of $x$, and hence contain a word of the form $xz$, must be positive. So $\kappa_c$ assigns $0$ to all classes $[x]_{{\simeq}_c}$ with $\pi_c^{\sim}(x) = 0$. Vice versa, if no negative idempotent classes are reachable from the class of $x$, then $(xz)^\omega \in L_c$ for all $z \in \Sigma^*$ with $xz \in E_c$ (if $(xz)^\omega \not\in L_c$, then it has an idempotent $(xz)^n$ as prefix by \cref{lem:idempotent}, which would then be negative).

	Now let $i \ge 1$ and assume by induction that $\kappa_c$ and $\pi_c^{\sim}$ coincide for all values less than~$i$. The proof is a simple reasoning using the definitions of $\kappa_c$ and $\pi_c^{\sim}$.

	Let $x \in \Sigma^+$. If  $\pi_c^{\sim}(x) = i$, then we already know by induction that $\kappa_c(x) \ge i$. We show that  $\kappa_c(x) = i$.
	Consider an idempotent class that is reachable from the class of $x$. This class is of the form $[xz]_{{\simeq}_c}$ for some $z$ with $xz \in E_c$.
	By definition, $\pi_c^{\sim}(x) = i$ implies that for each $z \in \Sigma^*$ with $xz \in E_c$, either $(xz)^\omega$ has a prefix in $P_{c,i-1}^{\sim}(L)$, or ($(xz)^\omega \in L_c$ if and only if $i$ is even).
	In the first case, if $(xz)^\omega$ has a prefix in $P_{c,i-1}^{\sim}(L)$, then there is an $n$ such that $\pi_c^{\sim}((xz)^n) \le i-1$ (because  $\pi_c^{\sim}$ is weak).
	Then by induction also $\kappa_c((xz)^n) \le i-1$.  Since $xz$ is idempotent, we have $xz \mathrel{{\simeq}_c} (xz)^n$, and thus also $\kappa_c(xz) \le i-1$. In the second case, ($(xz)^\omega \in L_c$ if and only if $i$ is even), the class of $xz$ is positive if and only if $i$ is even. Hence, each idempotent class that is reachable from $x$ either has a $\kappa_c$ value smaller than $i$ or is positive if and only if $i$ is even. Hence $\kappa_c(x) = i$.

	For the other direction, assume that $\kappa_c(x) = i$ and show that $\pi_c^{\sim}(x) = i$. By induction, $\pi_c^{\sim}(x) \ge i$, which means that $x \notin P_{c,i-1}^{\sim}$. Let $z \in \Sigma^*$ be such that $xz \in E_c$. Let $n$ be such that the class of $(xz)^n$ is idempotent (by \cref{lem:idempotent}). Since $\kappa_c(x) = i$, the class of $(xz)^n$ has a $\kappa_c$-value less than $i$, or it is positive if and only if $i$ is even.
	In the first case, we get by induction that $(xz)^n \in P_{c,i-1}^{\sim}$. In the second case, $((xz)^n)^\omega = (xz)^\omega \in L_c$ if and only if $i$ is even. So $x \in P_{c,i}^{\sim}$ and hence $\pi_c^{\sim}(x) = i$.
\end{proof}

If there are a negative and a positive idempotent class in the same SCC of a PRC, then none of them will satisfy the conditions for being assigned value $i$ by $\kappa_c$. This contradicts the fact that every class is assigned a value as a consequence of \cref{PRCcoloring}. We use this structural property of the PRCs in \cref{dpalearner}, and hence state it as a lemma.
\begin{lemma} \label{pureSCC}
	If two idempotent classes in a PRC of a canonical FORC are in the same SCC, then both are positive or both are negative. Furthermore, two classes in the same SCC of ${\simeq}_c$ have the same $\kappa_c$-value.
\end{lemma}
As a consequence of \cref{PRCcoloring}, we obtain a construction from canonical FORCs to DPAs that is only exponential in the parity complexity of the language, improving the upper bound from \cite[Proposition~5.9]{AngluinBF16} that is exponential in the size of the FORC (the statement in \cite{AngluinBF16} is about saturated families of DFAs, which are basically refinements of canonical FORCs up to some small details).
\begin{theorem}\label{FORCtoDPA}
	Let $L$ be a regular $\omega$-language of parity complexity $k$, ${\sim}$ be a right congruence that refines ${\sim}_L$, and $({\sim},\familyfor{\approx}{{\sim}})$ be the canonical FORC for $L$ and ${\sim}$. Then the precise DPA $\mc{A}_L^{\sim}$ for $L$ and ${\sim}$ has at most $n^k$ states, where $n$ is the size of the FORC.
\end{theorem}
\begin{proof}
	The Mealy machine obtained from the PRC ${\simeq}_c$ by using the transition structure of ${\simeq}_c$ and assigning the $\kappa_c$-value of the target state to a transition computes $\pi_c^{\sim}$ according to  \cref{PRCcoloring}. Denote the resulting family of Mealy machines by $\family\MM$. Then $\bow\AA(\family\MM)$ computes the priority mapping $\bow\pi^{\sim}$ and thus accepts $L$ by \cref{lem:joinlanguage}. Since each ${\simeq}_c$ refines ${\sim}$ by definition, the component for ${\sim}$ in the construction of $\bow\AA(\family\MM)$ is not required (and even if it is used it does not blow up the state space).
	The states used for the component of priority $i$ in $\bow\AA(\family\MM)$ is the sum of the states of the $\DD_{c,i}$, each of which has size bounded by the size of ${\simeq}_c$. So the number of states used for the component of priority $i$ is at most $n$.
	Hence, the size of $\bow\AA(\family\MM)$ is bounded by $n^k$.
\end{proof}

Our completeness result for the learning algorithm uses the fact that we can infer the syntactic FORC of a language if the sample contains enough information. For bounding the number of required examples, we give an upper bound on the size of the syntactic FORC for $L$ in the size of a DPA for $L$.
\begin{proposition} \label{FORCsize}
	Let $\AA$ be a normalized DPA with priorities $\{0,\dotsc,k-1\}$ recognizing the language $L \subseteq \Sigma^\omega$ and ${\sim}$ be a right congruence with ${\sim}_\AA \refine {\sim} \refine {\sim}_L$. Let $m$ be the number of ${\sim}$-classes, and $d$ such that for each class $c$ of ${\sim}$ there are at most $d$ many $c$-states in $\AA$. Then for the canonical FORC of $L$ with LRC ${\sim}$, the size of each PRC is at most $m(dk)^d$.
\end{proposition}
\begin{proof}
	Let $\mathcal{A} = (Q,\Sigma,q_0,\delta,\kappa)$, let $c$ be a class of ${\sim}$, and let ${\simeq}_c$ be the canonical PRC for $c$. Let $Q_c \subseteq Q$ be the set of $c$-states of $\AA$. Then by assumption $|Q_c|\le d$. For $x \in \Sigma^*$ define the function $\Delta_{c,x}: Q_c \rightarrow Q \times \{0,\ldots,k-1\}$ by induction on the length of $x$ by $\Delta_{c,\varepsilon}(q) = (q,k-1)$ and for $x = x'a$ with $\Delta_{c,x'}(q) = (q',i)$ let $\Delta_{c,x'a}(q) = (\delta(q',a),\min(i,\kappa(q,a)))$.

	If two words $x,y \in \Sigma^+$ with $ux \mathrel{\sim} uy$ define the same function $\Delta_{c,x} = \Delta_{c,y}$, then $x \mathrel{{\simeq}_c} y$. To see this, let $u \in c$, and $z \in \Sigma^*$ such that $uxz \mathrel{\sim} u$. For verifying $ux \mathrel{\sim} uy$, let $q := \delta^*(q_0,u) \in Q_c$. By assumption $\Delta_{c,x}(q) = \Delta_{c,y}(q) =: (q',i)$. Thus, $\delta^*(q_0,ux) = \delta^*(q_0,uy) = q'$, and hence $ux \mathrel{\sim} uy$.

	For verifying $(u(xz)^\omega \in L \Leftrightarrow u(yz)^\omega \in L)$, consider the runs of $\mathcal{A}$ on the two words. On $u$ they are the same. On $x$ and $y$ they might differ, but they reach the same state and see the same minimal priority on the way (because $\Delta_{c,x} = \Delta_{c,y}$). On $z$ they agree again, so they are in the same state after $uxz$ and $uyz$. Since $uxz \mathrel{\sim} u$, this state is in $Q_c$. We can now repeat this argument, so the two runs only differ on the parts resulting from the $x$, respectively $y$, segments of the word. But on these segments, the same minimal priority is visited. Hence, the minimal priority that appears infinitely often is the same in the two runs.

	Since ${\sim}_\AA$ refines ${\sim}$, $\Delta_{c,x}(q_1)$ and $\Delta_{c,x}(q_2)$ are ${\sim}$-equivalent for all $q_1,q_2 \in Q_c$. Hence, all the states in the image of $\Delta_{c,x}$ are ${\sim}$-equivalent. This implies that there are at most $m(dk)^d$ different possible functions for $\Delta_{c,x}$. Hence, the number of classes in ${\simeq}_c$ on $\Sigma^+$ is bounded by $m(dk)^d$.
\end{proof}

\section{DPA learner}\label{dpalearner}

In this section, we introduce the passive learner \dpainf (for ``DPA inference''), which constructs in polynomial time a DPA that is consistent with a given input sample.
Furthermore, the learner is designed in such a way that it can infer a DPA for every regular $\omega$-language $L$, given that the sample contains enough information on $L$ (which is made precise in the proof of \cref{dpainfcomplete}).

Throughout this section, we assume $S = (S_+, S_-)$ is a finite $\omega$-sample of ultimately periodic words.
The overall idea of the algorithm is to first infer a FORC from $S$ (step 1), color it according to the procedure outlined in \cref{forcs} (step 2),
and then construct the precise DPA from it as in \cref{FORCtoDPA}. However, the upper bound on the size of the precise DPA is exponential (in the number of priorities) compared to the size of the FORC, and further, even if the precise DPA is small, the construction might produce a large intermediate result before minimizing the DPA as a Mealy machine that emits priorities. For this reason, the step that transforms the coloring on the FORC into a DPA is split up into two parts (steps 3 and 4 of the algorithm). In step 3, the coloring of the precise DPA on prefixes of the sample is computed. This can be done in polynomial time thanks to \cref{lem:polynomialjoincoloring}. In step 4 an active learner for Mealy machines is used, where the oracle answers queries based on the coloring computed in step 3. This ensures that no large intermediate result is produced and that the computation remains polynomial in the size of the sample.

Note that these descriptions are given under the assumption that the sample contains enough information from an underlying language $L$ for inferring the correct objects in each step.
However, the algorithm has to produce some DPA that is consistent with the sample in polynomial time, even if the sample does not fully characterize the objects required by our algorithm.
In such situations, it might happen that intermediate results are inconsistent with the sample.
For example, in step 1 the algorithm might not be able to infer a FORC that is consistent with the sample, or in step 3, the coloring on the prefixes of the sample might not be consistent with the sample.
In such cases, the algorithm defaults, for that step, to a simple structure that is guaranteed to be consistent with the sample, such that it can proceed with the next step.
In such situations, one could also immediately return some default DPA that is consistent with the sample.
We decided to return default structures for the corresponding step because then the algorithm still might generalize from the sample in the following steps.

We later show that if $S$ contains enough information on the language $L$, then the syntactic FORC and thus the precise coloring will be inferred, and that the construction of the DPA from the coloring will result in a correct DPA for $L$ whose size is bounded by the precise DPA (\cref{dpainfcomplete}).

The theoretical foundations for the steps 2,3,4 have been developed in the previous sections. We now give more details on the first step, and then describe the learning algorithm.

\begin{algorithm}[bt]
	\caption{$\op{GLeRC}$\label{algo:glerc}}
	\DontPrintSemicolon
	\SetNoFillComment
	\KwIn{A function \op{cons} and a complete TS $\TT_0$, where\\
    \quad \op{cons} maps a partial TS to a boolean\\
    \quad $\TT_0$ serves as default and verifies $\op{cons}(\TT_0) = \true$.}
	\KwOut{A complete TS $\TT$ with $\op{cons}(\TT) = \true$ and $|\TT| \leq |\TT_0|$.}
	$Q \gets \{\varepsilon\}, \delta \gets \emptyset, \TT \gets (Q, \Sigma, \delta, \varepsilon)$\;
	\While{$\TT$ is not complete}{
	\lIf{$|\TT| > |\TT_0|$}{
		\Return{$\TT_0$}
	}
	$pa \gets$ llex.~minimal word in $Q\Sigma$ such that $\delta(p, a)$ is undefined\;
	\ForEach{$q \in Q$ in length-lexicographic order}{
		$\delta' \gets \delta \cup \{p \xrightarrow{a} q\}$\;
		\uIf{$\op{cons}((Q, \Sigma, \varepsilon, \delta'))$ returns \true}{
				$\delta \gets \delta'$\;
                \text{\textbf{continue} with next iteration of while loop}
			}
		}
		$Q \gets Q \cup \{pa\}, \delta \gets \delta \cup \{p \xrightarrow{a} pa\}$\;
	}
	\Return{$\TT$}
\end{algorithm}
For learning the FORC, we use the procedure \emph{G}reedily \emph{Le}arn \emph{R}ight \emph{C}ongruence (\glerc), which can be seen in \cref{algo:glerc}.
It receives as input a consistency function \op{cons} and incrementally constructs a TS $\TT$ by adding missing transitions, first trying existing states as targets in canonical order.
After inserting a transition to a potential target, \glerc verifies whether \op{cons} is satisfied.
If yes, \glerc keeps the transition and proceeds to the next iteration.
Otherwise, if no existing state works, \glerc adds a new state as target for the transition.
To ensure termination, \glerc has a second input, which is a fallback transition system $\TT_0$. If the constructed TS $\TT$ exceeds $\TT_0$ in size at any point, the algorithm terminates and returns $\TT_0$.
This guarantees runtime that is polynomial in the execution time of \op{cons} and in the size of $\TT_0$, since \glerc attempts to insert at most polynomially many transitions.

\begin{restatable}{proposition}{rstglerc}\label{glerc}
	For a consistency function \op{cons} and a complete transition system $\TT_0$ with $\op{cons}(\TT_0) = \true$, the algorithm \glerc computes a complete TS $\TT$ such that $\op{cons}(\TT) = \true$ in time polynomial in the runtime of \op{cons} and the size of $\TT_0$.
\end{restatable}

Before we go into the details of how we instantiate \glerc in our setting for learning the congruences of a FORC, let us briefly comment on the default transition system that is used for ensuring termination.
The algorithm \glerc has emerged as an abstract version of the learning algorithm \textsf{Sprout} for deterministic $\omega$-automata presented in  \cite{BohnL21}, which itself can be seen as an extension of the RPNI algorithm for learning DFAs from samples of finite words \cite{rpniOG}.
The instantiation of \glerc that corresponds to the RPNI algorithm for samples of finite words does not require a default transition system, because after polynomially many steps all finite words from the sample do have a path that completely lies inside the constructed transition system, and after that point no new states will be created by \glerc.
In the context of $\omega$-words, however, this is not necessarily the case: It was established in \cite[Proposition~10]{BohnL21} that there exist samples for which the \textsf{Sprout} algorithm does not terminate without the use of a default transition system.
While this example does not directly apply to the instantiations of \glerc that we use below, there are no easy termination arguments in the setting of $\omega$-words as it is the case for finite words.
For this reason, we use the default transition system. This allows us to cover existing algorithms in a unified form, and it gives an easy termination argument.

Let us now describe how we instantiate \glerc for learning the congruences of a FORC. The \op{cons} functions that we pass to \glerc, verify that the constructed TS satisfies certain notions of consistency with the sample, which relativize the definitions of the canonical FORC and its coloring from \cref{forcs} to a sample $S$ (\cref{def:consistencies}).
In \cref{mnconsistent} and \cref{iterationconsistent} we show that these consistency checks can be performed in polynomial time. We then combine these two consistency notions into a notion of consistency of a sample with a FORC (\cref{forcconsistent}), followed by a short description of the default transition system that we use in combination with the two consistency functions. This allows us to subsequently use the \glerc algorithms with different \op{cons} functions for inferring the leading and progress right congruences in step~1 of the learning algorithm.
As we see later in the proof of \cref{dpainfcomplete}, for these consistency functions there are simple arguments for showing that the appropriate transition systems are inferred in the limit.

\begin{definition}\label{def:consistencies}
	Let $S = (S_+, S_-)$ be an $\omega$-sample and $\TT$ be a partial transition system. Two words $x, y \in \Sigma^*$ \emph{are separated by} $\TT$, if $\TT(x)$ or $\TT(y)$ is undefined, or if $\TT(x) \neq \TT(y)$.
	\begin{itemize}
		\item We call $\TT$ \emph{MN-consistent} with $S$, if for all $xuv^\omega \in S_+, yuv^\omega \in S_-$, the prefixes $x$ and $y$ are separated by $\TT$, and
		\item say that $\TT$ is \emph{iteration consistent} with $(S, {\sim}, c)$, where ${\sim}$ represents a complete transition system that is MN-consistent with $S$ and $c$ is a class of ${\sim}$, if for all $xz, yz \in E_c^{\sim}$ with $(xz)^\omega \in S_+, (yz)^\omega \in S_-$ holds that $x$ and $y$ are separated by $\TT$
	\end{itemize}
\end{definition}

A check for MN-consistency is part of the consistency checks in the learning algorithms from  \cite{BohnL21}. For making the paper self-contained, we nevertheless describe a possible algorithm below.
For showing that MN-~and iteration consistency can be verified in polynomial time, we reduce the problems to a unified setting, which can be solved by performing a polynomial number of reachability analyses.
Specifically, we construct two DFAs $\AA_1, \AA_2$ and a conflict relation $C \subseteq Q_1 \times Q_2$ from the sample, to encode the pairs of words $x, y$ which are not allowed to lead to the same state in a consistent transition system $\TT$.
Now we can check for consistency of $\TT$ by verifying for all $x, y \in \Sigma^*$, that $(\delta^*_1(\iota_1, x), \delta^*_2(\iota_2, y)) \in C$ implies that $x$ and $y$ do not lead to the same state in $\TT$.
In other words $\TT$ is not consistent if and only if there is some $q$ in $\TT$ such that $(q,q_1)$ and $(q,q_2)$ are reachable in $\TT \times \AA_1$ and $\TT \times \AA_2$, respectively, and $(q_1,q_2) \in C$.
This can clearly be checked in time polynomial in the size of $|\TT|$, $|\AA_1|$ and $|\AA_2|$.
In the proof of the following two lemmas, we thus only need to show that suitable DFAs $\AA_1, \AA_2$ and a conflict relation $C$ can be constructed in polynomial time.

\begin{lemma}\label{mnconsistent}
	It can be verified in polynomial time whether a (partial) transition system $\TT$ is MN-consistent with an $\omega$-sample $S$.
\end{lemma}
\begin{proof}
	We construct DFAs $\AA_1, \AA_2$ such that $\AA_1$ and $\AA_2$ accept all prefixes of $S_+$ and $S_-$, respectively. This can easily be done by using the prefix trees for $S_+$, respectively $S_-$, and attaching a loop for the periodic part once the prefix uniquely identifies the example. All states of the resulting transition system are accepting.  Missing transitions are directed into a rejecting sink state.

	A pair $(q_1, q_2)$ of states with $q_i \in \AA_i$ for $i \in \{1,2\}$ is in the conflict relation $C$, if there exists an infinite run in the product automaton $\AA_1 \times \AA_2$, which starts in $(q_1, q_2)$ and stays in $F_1 \times F_2$.
	Checking for the existence of such an infinite run starting in a pair $(q_1, q_2)$ can be done by restricting the product to $F_1 \times F_2$ and searching for a loop.
	As the product $\AA_1 \times \AA_2$ is polynomial in $S$, the construction of $C$ is clearly possible in polynomial time.

	To see that the construction is correct, let $q$ be a state of $\TT$ such that $(q,q_1)$ and $(q,q_2)$ are reachable in $\TT \times \AA_1$ and $\TT \times \AA_2$, respectively, and $(q_1,q_2) \in C$.
	Further, let $r_1, r_2 \in \Sigma^*$ be such that $\TT \times \AA_i$ reaches $(q, q_i)$ when reading $r_i$ from the initial state.
	As $(q_1, q_2) \in C$, there exists an infinite run from $(q_1,q_2)$ in $\AA_1 \times \AA_2$, which stays in $F_1 \times F_2$.
	This means for some $x \in \Sigma^*, y \in \Sigma^+$ we have $(q_1,q_2) \xrightarrow{x} (p_1,p_2) \xrightarrow{y} (p_1,p_2)$ while only visiting states from $F_1 \times F_2$.
	Then $r_1xy^\omega \in \sem{S_{+}}$ and $r_2xy^\omega \in \sem{S_-}$ and because $r_1$ and $r_2$ reach the same state in $\TT$, $\TT$ cannot be MN-consistent with $S$.

	For the other direction, assume that $\TT$ is not MN-consistent with $S$.
	This means there exist $xw, yw \in \sem{S}$ with $xw \in \sem{S_+}$ and $yw \in \sem{S_-}$, but $x$ and $y$ lead to the same state $q$ in $\TT$.
	The pair $(q_1, q_2) := (\delta^*_1(\iota_1, x), \delta^*_2(\iota_2, y))$ is in $C$, because for each prefix $z$ of $w$, $xz, yz \in \prf{\sem{S}}$, guaranteeing that there exists an infinite run that starts in $(q_1,q_2)$ and remains in $F_1 \times F_2$.

	Overall, $\TT$ is \emph{not MN-consistent} with $S$, if and only if there is some $q$ in $\TT$ such that $(q,q_1)$ and $(q,q_2)$ are reachable in $\TT \times \AA_1$ and $\TT \times \AA_2$, respectively, and $(q_1,q_2) \in C$.
	As the sizes of $\AA_1, \AA_2$ and $C$ are polynomial in $|S|$, it follows from the remarks preceding this proof that MN-consistency can be checked in polynomial time.
\end{proof}

\begin{lemma}\label{iterationconsistent}
	For a given finite $\omega$-sample $S$, a right congruence ${\sim}$ whose TS is MN-consistent with $S$, a class $c$ of ${\sim}$, and a partial TS $\TT$, it can be checked in polynomial time whether $\TT$ is iteration consistent with $(S, {\sim}, c)$.
\end{lemma}
\begin{proof}
	The proof of this lemma is analogous to the preceding \cref{mnconsistent}, but differs in the construction of the $\AA_1, \AA_2$ and the conflict relation.
	We build the DFAs such that
	\begin{itemize}
		\item $\AA_1$ accepts a word $x$ if $x \in E_c^{\sim}$ and $x^\omega \in S_{+}$, whereas
		\item $\AA_2$ accepts $x$ if $x \in E_c^{\sim}$ and $x^\omega \in S_{-}$.
	\end{itemize}
	For the construction of the $\AA_i$, let $\check{R}_{\sigma}$ for $\sigma \in \{+,-\}$ be the set consisting of all $y$ such that $y$ is a shortest word with $y^\omega \in \sem{S_{\sigma}}$. For computing the sets $\check{R}_{\sigma}$, one can consider each $uv^\omega$ in $S$, check whether it is periodic, and if yes, compute the shortest $y$ with $y^\omega = uv^\omega$. This can be done, for example, by building a DFA that accepts precisely the prefixes of $uv^\omega$ (using the prefixes of $uv$ as states, identifying the states for $u$ and $uv$). This DFA can be minimized, and then $uv^\omega$ is periodic if and only if the minimal DFA consists of a loop of accepting states (with one rejecting sink). The label sequence of this loop gives the shortest word $y$. One can also use other methods for computing $\check{R}_{\sigma}$, for example the characterization of minimal representations of ultimately periodic words from \cite[Proposition~4.42]{Landwehr21}.

	Now it is not hard to see that the words $x$ with $x^\omega \in S_\sigma$ are precisely those of the form $y^n$ with $y \in \check{R}_{\sigma}$. To see that, note that $x^\omega \in S_\sigma$ implies that $x^\omega = y^\omega$ for some $y \in \check{R}_{\sigma}$. So $x$ must be of the form $y^nz$ with $z$ a strict prefix of $y$. If $z \not= \varepsilon$, this would imply that $z^\omega = y^\omega$, contradicting the fact that $y$ is the shortest word representing $y^\omega$.

	With these observations, it is then easy to build the DFAs $\AA_1,\AA_2$: The states are prefixes of $y^\omega$ for $y \in \check{R}_{\sigma}$ (with $\sigma = +$ for $\AA_1$, and $\sigma = -$ for $\AA_2$). For the least $k$ such that $y^k$ is not a prefix of another $(y')^\omega$ anymore, start looping on $y$.
	Finally, we intersect the resulting automata with DFAs for $E_c^{\sim}$, which can be directly obtained from the transition structure of ${\sim}$ by using $c$ as initial and final state (since the DFAs that we built from $\check{R}_{\sigma}$ only accept non-empty words, we do not need to exclude the empty word in the DFA for $E_c^{\sim}$).
	Overall, the size of the resulting automata is clearly polynomial in $|S|$.

	The conflict relation is the least relation $C$ over $Q_1 \times Q_2$ containing $F_1 \times F_2$ and verifying that $(\delta_1(q_1, a), \delta_2(q_2, a)) \in C$ implies $(q_1, q_2) \in C$.
	The relation $C$ can be obtained by a fixpoint iteration, ensuring that the number of iterations is bounded by $|Q_1 \times Q_2|$.
	Thus, the computation of $C$ is possible in polynomial time.

	For the correctness of the construction, let $q$ be a state of a partial TS $\TT$ which is MN-consistent with $S$.
	Further, let $(q,q_1)$ and $(q, q_2)$ with $(q_1, q_2) \in C$ be reachable on $r_1$ in $\TT \times \AA_1$ and on $r_2$ in $\TT \times \AA_2$, respectively.
	As $(q_1, q_2) \in C$, there must exist some word $x \in \Sigma^*$ such that $(p_1,p_2) := (\delta^*_1(q_1,x), \delta^*_2(q_2,x)) \in F_1 \times F_2$.
	This means $r_1x, r_2x \in E_c^{\sim}$ and $(r_1x)^\omega, (r_2x)^\omega \in \sem{S}$ with $((r_1x)^\omega \in \sem{S_{+}} \iff (r_2x)^\omega \in \sem{S_{-}})$.
	But then because $r_1$ and $r_2$ reach the same state in $\TT$, it follows that $\TT$ is not iteration consistent with $(S, {\sim}, c)$.

	For the other direction, assume that $\TT$ is MN-consistent with $S$ but not iteration consistent with $(S, {\sim}, c)$.
	This means we can find $xz, yz \in E^{\sim}_c$ such that $(xz)^\omega \in \sem{S_+}$, $(yz)^\omega \in \sem{S_-}$ and $x,y$ lead to the same state $p$ in $\TT$.
	Clearly it holds that $(q_1, q_2) := (\delta^*_1(xz), \delta^*_2(yz)) \in F_1 \times F_2 \subseteq C$.
	Moreover, our definition of $C$ ensures that removing the common suffix $z$ from both words retains membership in $C$, meaning $(\delta^*_1(x), \delta^*_2(y)) \in C$.

	Overall, we obtain that $\TT$ is \emph{not} iteration consistent with $(S, {\sim}, c)$, if and only if there is some $q$ in $\TT$ such that $(q,q_1)$ and $(q,q_2)$ are reachable in $\TT \times \AA_1$ and $\TT \times \AA_2$, respectively, and $(q_1,q_2) \in C$.
	As $\AA_1, \AA_2$ and $C$ can be constructed in polynomial time, checking for iteration consistency effectively reduces to computing the reachable states in the two products $\TT\times \AA_1$ and $\TT \times \AA_2$ and checking whether a pair with the properties outlined above exists.
	This is clearly possible in time polynomial in the size of $|\TT|$, $|\AA_1|$ and $|\AA_2|$, thus concluding the proof.
\end{proof}

We now combine the notions of MN-~ and iteration consistency to define the conditions under which a FORC is consistent with a sample.

\begin{definition}\label{forcconsistent}
	A FORC $({\sim}, (\approx_c)_{c \in [{\sim}]})$ is \emph{consistent with $S$} if
	\begin{itemize}
		\item ${\sim}$ is MN-consistent with $S$ and
		\item for each $c \in [{\sim}]$, $\approx_c$ is MN-consistent with $S_c$ and iteration consistent with $(S_c,{\sim},c)$, where $S_c = (S_{c,+}, S_{c,-})$ and $S_{c,\sigma} = \{uv^\omega \mid \exists x \in c$ with $xuv^\omega \in \sem{S_\sigma}\}$ for $\sigma \in \{+,-\}$.
	\end{itemize}
\end{definition}

\begin{figure}
	\tikzstyle{every state}=[draw=none]
	\begin{center}
          	\includegraphics[width=0.3\textwidth]{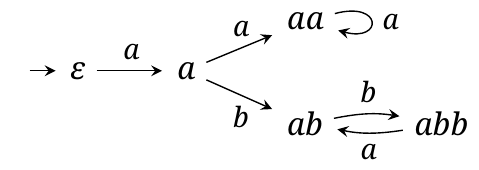}
	\end{center}
	\caption{
		An example for the construction of a default structure ${\sim}_S$ for the sample $S$ containing only the examples $a^\omega$ and $ab(ba)^\omega$. Note, that the classification of these words is not relevant for the definition of the default structure ${\sim}_S$, which is why we omit it from this description.
	}
	\label{fig:default}
\end{figure}

In the following, we outline how default structures that are MN-~ and iteration consistent with a given sample $S$ can be constructed.
An illustration for this construction can be seen in \cref{fig:default}.
For a given sample $S$, we can construct default structures that are MN- and iteration consistent as follows.
Consider the prefix tree of $S$, to which we add one sink state for words that are not prefixes of $S$, and loops on the shortest words which are prefix of exactly one example word.
Formally, we define ${\sim}_S$ as $x \mathrel{{\sim}_S} y$ if $ x,y \notin \prf{(\sem{S})}$ or we have $x^{-1}\sem{S} = y^{-1}\sem{S}$ and $|x^{-1}\sem{S}| = 1$.
It is not hard to see that the size of ${\sim}_S$ is polynomial in $|S|$, see for example \cite[Appendix~B]{BohnL21arxiv}.
Further, $\sim_S$ is MN-~ and iteration consistent with  $S$ because it separates all words $x,y$ that are prefixes of different words in $S$, which directly implies that the conditions from \cref{def:consistencies} are satisfied.

\paragraph*{Formal description of the learner}
The learner that we propose consists of several steps (which are sketched at beginning of this \cref{dpalearner}). An overview of the algorithm can be found in \cref{fig:dpainf}. Roughly speaking, it extracts a FORC from the given sample, colors it using the algorithm from \cref{forcs}, builds a family of Mealy machines for weak priority mappings capturing the periodic part of the sample, and then uses an active learner for Mealy machines for joining the learned FWPM.
In the following, we explain each step of the learner \dpainf in more detail.
The full arguments why this results in a polynomial time consistent learner that can learn every regular $\omega$-language in the limit is given in the proofs of \cref{dpainfconsistent} and \cref{dpainfcomplete}. In the description of the steps we only indicate some of these arguments in order to ease the understanding of the algorithm.

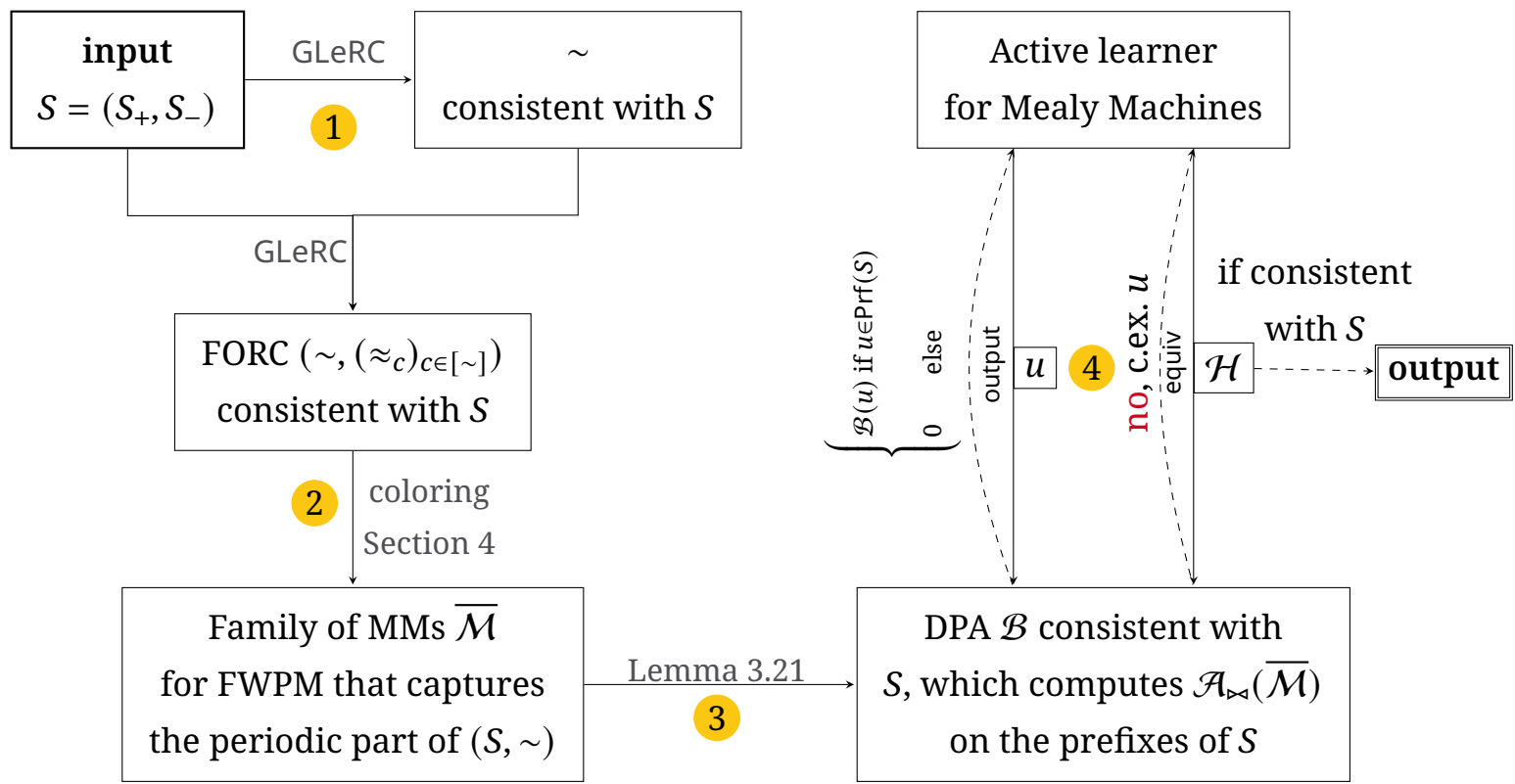
\begin{figure}
	\begin{tikzpicture}
		{
			\tikzset{%
				every node/.style={draw, rectangle}
			}
			\node[thick,anchor=west] (sample) {\begin{tabular}{c}\textbf{input}\\$S = (S_+, S_-)$\end{tabular}};
   
			\node (sim) at ($(sample)+(0.3\textwidth,0)$) {\begin{tabular}{c}${\sim}$\\consistent with $S$\end{tabular}};
   
			\node (forc) at ($(sample)!0.5!(sim) + (0,-0.14\textheight)$) {\begin{tabular}{c}FORC $({\sim},(\approx_c)_{c \in [{\sim}]})$\\consistent with $S$\end{tabular}};
   
			\node (mealies) at ($(forc)+(0,-0.14\textheight)$) {\begin{tabular}{c}Family of MMs $\overline\MM$\\for FWPM that captures\\the periodic part of $(S, {\sim})$\end{tabular}};
   
			\node (dpab) at ($(mealies)+(0.5\textwidth,0)$) {\begin{tabular}{c}DPA $\BB$ consistent with\\$S$, which computes $\bow\AA(\overline\MM)$\\on the prefixes of $S$\end{tabular}};
   
			\node (al) at (sample -| dpab) {\begin{tabular}{c}Active learner\\for Mealy Machines\end{tabular}};
   
			\node[anchor=east, double,yshift= 0.1\ThCScm] (result) at (forc -| {$(sample.west)+(\textwidth,0)$}) {\textbf{output}};
		}

		\draw[->]
		(sample.east) -- (sim.west);
		\draw[->] (sample.south) -- ++(0,-0.7\ThCScm) -- ({$(sample.south) + (0,-0.7\ThCScm)$} -| {forc.north}) -- (forc.north);
		\draw[->] (sim.south) -- ++(0,-0.7\ThCScm) -- ({$(sim.south) + (0,-0.7\ThCScm)$} -| {forc.north}) -- (forc.north);
		\draw[->] (forc.south) -- (mealies.north);
		\draw[->] (mealies.east) -- (dpab.west);

		\draw[->] ($(al.south)+(0.06\textwidth,0)$) to node[draw, rectangle] (hypothesis) {$\HH$} node[rotate=90, yshift=0.2\ThCScm,xshift=-0.4\ThCScm] {$\scriptstyle{\op{equiv}}$} (\tikztostart |- {dpab.north});
		\draw[<-, bend right=15, dashed] ($(al.south)+(0.06\textwidth,0)$) to node[rotate=90, xshift=-0.8\ThCScm, yshift=0.2\ThCScm] {\red{no}, c.ex. $u$} (\tikztostart |- {dpab.north});
		\draw[->, dashed] (hypothesis) edge node[yshift=0mm] {\begin{tabular}{c}if consistent\\with $S$\end{tabular}} (result);

		\draw[->] ($(al.south)+(-0.06\textwidth,0)$) to node[rotate=90, xshift=-0.5\ThCScm, yshift=0.2\ThCScm] {$\scriptstyle{\op{output}}$} node[draw, rectangle] (wordu) {$u$} (\tikztostart |- {dpab.north});
		\draw[<-, bend right=20, dashed] ($(al.south)+(-0.06\textwidth,0)$) to node[rotate=90,xshift=-1.1\ThCScm,yshift=0.8\ThCScm] {%
			$\begin{cases}\scriptstyle{\BB(u) \text{ if } u \in \prf(S)} \\
					\scriptstyle{0 \quad\phantom{i}\text{ else}}\end{cases}$%
		} (\tikztostart |- {dpab.north});

		{
		\tikzset{every node/.style={lipicsYellow, fill=lipicsYellow, draw, circle, text=black, inner sep=2pt}}
		\node (step1) at ($(sample)!0.5!(sim) + (-0.2\ThCScm,-0.5\ThCScm)$) {1};
		\node (step2) at ($(forc)!0.4!(mealies) + (-0.4\ThCScm,0)$) {2};
		\node (step3) at ($(mealies)!0.485!(dpab) + (0,-0.35\ThCScm)$) {3};
		\node (step4) at ($(wordu)!0.3!(hypothesis)$) {4};
		}
		{
		\tikzset{every node/.style={lipicsGray, draw=none, font=\small}}
		\node (glerc1) at ($(sample)!0.5!(sim) + (-0.15\ThCScm,0.3\ThCScm)$) {\hyperref[algo:glerc]{\glerc}};
		\node (glerc2) at ($(glerc1) + (-0.4\ThCScm,-2.1\ThCScm)$) {\hyperref[algo:glerc]{\glerc}};
		\node (coloring) at ($(step2) + (1.2\ThCScm,-0.15\ThCScm)$) {\begin{tabular}{c}\hyperref[kappacoloring]{coloring}\\\cref{forcs}\end{tabular}};
		\node (todpa) at ($(step3) + (0,0.5\ThCScm)$) {\cref{lem:polynomialjoincoloring}};
		}
	\end{tikzpicture}
	\caption{Schematic view of the \dpainf algorithm, details for each step are in the text.}
	\label{fig:dpainf}
\end{figure} 

\subparagraph*{Step 1}
In the first step, \dpainf learns a FORC that is consistent with the sample $S$ using \glerc. The procedures for checking MN-consistency and iteration consistency are based on \cref{mnconsistent} and \cref{iterationconsistent}.
For the LRC, \dpainf calls \glerc with ${\sim}_S$ as default and using $\op{LRC-cons}$ as consistency function, where $\op{LRC-cons}(T)$ returns \true if and only if $T$ is MN-consistent with $S$ (see \cref{mnconsistent}).
We denote the resulting right congruence with ${\sim}$.
For each $c \in [{\sim}]$ and $\sigma \in \{+,-\}$, the learner now computes the sets
\begin{align*}
	S_{c,\sigma} & = \{uv^\omega \mid \exists x \in c \text{ with } xuv^\omega \in \sem{S_\sigma}\}
	\quad
	R_{c,\sigma} = \{v^\omega\mid \exists x \in c \text{ with }xv^\omega \in \sem{S_\sigma} \text{ and } xv \mathrel{\sim} x\}.
\end{align*}
It then executes \glerc with the product $({\sim}_{S_c} \times {\sim})$ as default and the consistency function $\op{PRC-cons}(\TT)$, which returns \true if
\begin{itemize}
	\item $\TT$ is MN-consistent and iteration consistent with $S_c$
	\item $\op{SCC}_\TT(q) \cap \op{SCC}_\TT(q') = \emptyset$ for all $q \in \inf_\TT(R_{c,+}), q' \in \inf_\TT(R_{c,-})$.
\end{itemize}
Thereby, a FORC $\FF = ({\sim}, (\approx_c)_{c \in [{\sim}]})$ that is consistent with $S$ is obtained.
The second condition ensures that $\FF$ meets the properties of \cref{pureSCC}. This makes it possible to compute a coloring analogous to $\kappa_c$, in the next step.

\subparagraph*{Step 2}
In the following, we use $\FF = ({\sim}, (\approx_c)_{c \in [{\sim}]})$ to refer to the FORC that is constructed in step 1.
The algorithm now builds a family of Mealy machines $\overline\MM = (\MM_c)_{c \in [{\sim}]}$ computing a weak priority mapping $\overline\gamma = (\gamma_c)_{c \in [{\sim}]}$.
This is done in a way which mimics the definition of the coloring $\kappa_c$ in \cref{forcs} and ensures that $\overline\gamma$ is the precise FWPM of $L$, provided the syntactic FORC was learned in step 1 (see \cref{PRCcoloring}).
For a class $c \in [{\sim}]$, we proceed as follows:
\begin{itemize}
	\item Label all states with $\sigma$ in $\TT_{\approx_c}$, the TS that represents $\approx_c$, that appear in $\inf_\TT(R_{c,\sigma})$ with $\sigma \in \{+,-\}$.
	\item For each state $q$, compute the set $P_q$ containing all states that are reachable from $q$ and have a classification.
	\item Starting with $i = 0$ and increasing $i$ after every iteration, assign $i$ to those states $q$ such that each $p \in P_q$ either already has a priority, or ($p$ has classification $+$ if and only if $i$ is even).
\end{itemize}
As the construction of $\FF$ in step 1 ensures that no SCCs with positive and negative looping sample words exist, each class of a PRC $\approx_c$ is assigned a priority.
For each class $c$ of ${\sim}$, we obtain a Mealy machine $\MM_c$ on the transition system of $\approx_c$, which uses the priority (with regard to the computed coloring) of the target state for each transition.
Overall, this procedure returns in polynomial time a family of Mealy machines $\overline\MM = (\MM_c)_{c \in [{\sim}]}$ whose size is polynomial in $\FF$.

\subparagraph*{Step 3}
In the following, we use $\bow\AA$ to refer to the DPA $\bow\AA(\overline\MM)$ from \cref{preciseDPA}.
Note that we cannot simply construct $\bow\AA(\overline\MM)$ because its size (or the size of an intermediate result) can be exponential in $\overline\MM$.
Instead, \dpainf builds a DPA $\BB$ which computes the priority mapping
\[
	\BB : u \mapsto \bow\AA(u) \text{ if $u \in \prf(\sem{S})$ and } u \mapsto 0 \text{ otherwise}.
\]
It can be shown by using \cref{lem:polynomialjoincoloring} that constructing $\BB$ is possible in polynomial time (see the proof of \cref{dpainfconsistent} below).
Then, our algorithm checks if $\BB$ is consistent with $S$, i.e.~that $S_+ \subseteq L(\BB)$ and $S_- \cap L(\BB) = \emptyset$.
If yes, the algorithm proceeds to step 4. Otherwise, we redefine $\BB$ as a fallback DPA that outputs $1$ on infixes of negative loops and $0$ otherwise.
Formally, the fallback DPA computes the priority mapping $\BB : \Sigma^+ \to \{0,1\}$ with $\BB(u) \mapsto 1$ iff $u \in \prf\big(\bigcup_{c \in [{\sim}]} R_{c,-}\big)$.
Note that this fallback DPA is consistent with $S$ since it outputs only priority $1$ on all negative examples, and finally only priority $0$ on all other $\omega$-words.
Hence, the DPA $\BB$ that is passed to the next step is consistent with $S$.

\subparagraph*{Step 4}
The purpose of this step is to compute a small DPA that is consistent with the sample and minimal viewed as a Mealy machine.
To achieve this, step~4 runs a polynomial time active learner \mealyal for Mealy machines as a black box, answering the queries based on the coloring computed in step 3. If the sample is complete for a language $L$, then step 3 terminates with a coloring on the prefixes of the sample that corresponds to the one computed by the precise DPA for $L$. If all queries posed by \mealyal can be answered using prefixes of the sample, we can prove that the precise DPA for $L$ (or a smaller one) will be learned.

The oracle that is provided to \mealyal uses the DPA $\BB$ from step 3 and reacts to the posed queries as follows:
\begin{align*}
	\op{output}(u)  & \rightsquigarrow \text{ answer } \BB(u)                                                                                                          \\
	\op{equiv}(\HH) & \rightsquigarrow \begin{cases}
		                                   \text{terminate and return $\HH$}                                          & \text{if $\HH$ is consistent with $S$} \\
		                                   \text{answer }\min_{\text{llex}} \{x \in \prf(S) \mid \BB(x) \neq \HH(x)\} & \text{otherwise.}
	                                   \end{cases}
\end{align*}

Note that, by definition, if this step terminates, then the resulting DPA is consistent with the sample.
Further, this step is guaranteed to terminate because the coloring from step 3 is consistent with the sample (see proof of \cref{dpainfconsistent}).

This concludes the description of the algorithm, and we now state the main results of this section.

\begin{theorem}\label{dpainfconsistent}
	\dpainf computes in polynomial time a DPA that is consistent with the $\omega$-sample $S$ it receives as input.
\end{theorem}
\begin{proof}
	Let $S = (S_+, S_-)$ be the sample on which \dpainf is called.
	By definition, step 4 always terminates with a hypothesis \(\HH\) that is consistent with \(S\).
	Thus, it remains to show that \dpainf always terminates in polynomial time.
	In the following, we consider each step of the algorithm individually:
	\begin{itemize}
		\item In \cref{mnconsistent} and \cref{iterationconsistent}, it was established that the two consistency functions \op{LRC-cons} and \op{PRC-cons}, which are passed to \glerc, run in polynomial time.
		      Thus, by \cref{glerc}, it follows that the first step terminates in polynomial time.

		\item The second step begins by computing the infinity sets of all sample words from $R_{c}$ in the respective PRC $\approx_c$ for a $c \in [{\sim}]$, which is possible in time polynomial in $|\approx_c|$ and $S$.
		      Further, the number of distinct priorities is bounded for each class $c \in [{\sim}]$ by the number of SCCs in $\approx_c$.
		      Therefore the computation of each $\kappa_c$ terminates after at most polynomially many iterations.
		      Since each iteration consists of a reachability analysis and some elementary operations, the construction of $\family\MM$ is clearly possible in polynomial time.

		\item We use $\bow\AA$ to denote $\bow\AA(\family\MM)$.
		      For the computation of $\BB$, we first build what we call a \emph{colored sample} $\hat{S}$ as follows.
		      For each sample word $uv^\omega \in \sem{S}$, we compute the sequence of colors produced by $\bow\AA$ on $uv^\omega$.
		      By \cref{lem:polynomialjoincoloring}, we can write this sequence as $rs^\omega$, where $|rs|$ is polynomial in $|uv|$ and $|\family\MM|$.
		      The product of $uv^\omega$ with $rs^\omega$ (which is a sequence of letter-priority pairs), is then polynomial in $|S|$.
		      Overall, the size of the colored sample $\hat{S}$ is polynomial in $|S|$.

		      $\BB$ then basically is the prefix tree of $\hat{S}$ starting to loop on the periodic part if the prefix uniquely identifies the word in $\hat{S}$, and a sink state $q_\bot$ for all non-prefixes of $S$.
		      More formally, for each colored example $\hat{u}\hat{v}^\omega$, there is $n$, polynomial in the size of $\hat{S}$, such that $\hat{u}\hat{v}^n$ is not a prefix of any other example in $\hat{S}$.
		      As states of $\BB$ we take all the prefixes of $\hat{u}\hat{v}^{n+1}$ (for all the examples from $\hat{S}$). For some prefix $\hat{x}$ and a letter $a$, we define the transitions $\delta_{\BB}(\hat{x},a)$ as follows. If for all $i$, $\hat{x}(a,i)$ is not a prefix of $\hat{S}$, then $\delta(\hat{x},a) = (q_\bot,0)$.
		      Otherwise, $\hat{x}(a,i)$ is a prefix of $\hat{S}$ for a unique $i$. If $\hat{x}(a,i)$ is a state of $\BB$, then $\delta_{\BB}(\hat{x},a) = (\hat{x}(a,i),i)$.
		      Otherwise, $\hat{x}$ is of the form $\hat{u}\hat{v}^{n+1}$, and $\hat{u}\hat{v}^{n+1}(a,i)$ is a prefix of $\hat{u}\hat{v}^\omega$, and only of that word.
		      We then set $\delta_{\BB}(\hat{x},a) = (\hat{u}\hat{v}^n(a,i),i)$, defining the loop for the periodic part of $\hat{u}\hat{v}^\omega$.
		      This construction results in a polynomial size DPA $\BB$ that computes the priority mapping defined in step 3.

		      It can be checked in polynomial time whether $\BB$ is consistent with the sample (actually, this can already be checked on $\hat{S}$ by looking at the minimal priorities on the loops of the colored examples).
		      In case $\BB$ is not consistent with $S$, \dpainf constructs a fallback DPA.
		      The construction of such a fallback DPA also works in polynomial time along the same lines as the construction explained above (take infixes of the periodic parts of negative example words as states, entering a loop once the infix uniquely identifies the periodic part of a negative example).

		\item Finally, in step 4, all answers provided to \mealyal are consistent with the DPA $\BB$ from step 3, which by the previous considerations is polynomial in $|S|$.
		      Since $\BB$ itself is consistent with $S$, it follows from the properties of \mealyal, that it terminates at the latest with the DPA $\BB$, or it terminates earlier with a strictly smaller DPA that is consistent with $S$.\\
		      To conclude that the execution of \mealyal indeed terminates in polynomial time, we need to verify that each counterexample provided by the teacher is polynomial in $|S|$.
		      So let $\HH$ be the hypothesis posed by \mealyal, which is not consistent with $S$.
		      This means $\HH$ and $\BB$ viewed as Mealy machines compute different priority mappings and a counterexample corresponds to a prefix $x$ of $S$ witnessing that $\HH(x) \neq \BB(x)$.
		      As all hypotheses given by \mealyal in an equivalence query have at most as many states as $\BB$, the shortest such $x$ is clearly polynomial in $S$.
	\end{itemize}
	As each step of \dpainf runs in polynomial time, and we have shown that the DPA it constructs must be consistent with $S$, the statement follows.
\end{proof}

The core idea for constructing a characteristic sample is to simulate a run \dpainf and whenever the algorithm would make a mistake in producing the canonical object we are interested in, we add words preventing it, which leads to the following theorem.

\begin{theorem}\label{dpainfcomplete}
	\dpainf can learn a DPA $\AA$ for every $\omega$-regular language $L$ in the limit. Moreover, the size of $\AA$ is bounded by the size of the precise DPA $\AA_L$ for $L$, and there exists a characteristic sample $S^L$ for $L$ that is polynomial in the maximum of the size of $\AA_L$ and the size of the syntactic FORC for $L$.
\end{theorem}
\begin{proof}
	Let $L$ be a regular $\omega$-language over some alphabet $\Sigma$.
	Recall that our goal is to construct a characteristic sample $S^L = (S_+, S_-)$ for $L$, meaning that \dpainf infers a DPA for $L$ from any sample $S'$ which contains $S^L$ and is consistent with $L$.
	In the following, we illustrate how the sample $S^L$ can be built in such a way that it guarantees that
	\begin{itemize}
		\item the syntactic FORC $\FF_L = ({\sim}_L, ({\simeq}_c)_{c \in [{\sim}]})$ for $L$ is learned in the first step,
		\item a family of Mealy machines $\family\MM = (\MM_c)_{c \in [{\sim}_L]}$ for the precise coloring $\family{\pi^{{\sim}_L}} = \familyfor\pi{{\sim}_L}$ is learned in step $2$,
		\item for all $\op{output}(u)$ queries posed by \mealyal in step 4, $u \in \prf(\sem{S^L})$ and
		\item it contains the necessary counterexamples to ensure that \mealyal terminates with a DPA for $L$.
	\end{itemize}
	For convenience, we do not describe $S^L$ directly.
	Instead, we build a set $W$ of ultimately periodic words and subsequently define $S^L = (W \cap L, W \setminus L)$.

	The core idea for ensuring that the syntactic FORC $\FF_L$ is learned in the first step, is to simulate the necessary executions of \glerc step by step, and to extend $W$ whenever \glerc would otherwise insert a transition that is not present in the target RC.
	Let $\TT$ be the partial TS that \glerc has constructed up to the point at which we want to prevent the insertion of a transition.
	For a state $q$ of $\TT$, we use $\minrep q$ to denote the length-lexicographically minimal word that reaches $q$.
	Then, if a transition $\delta(q,a) = p$ has to be prevented, it must be that $xa \mathrel{\not{\sim}} y$ for $x = \minrep{q}, y = \minrep{p}$ because up to now all transitions have been inserted correctly.
	Depending on the kind of right congruence we want to learn, we add the following words to $W$:
	\begin{itemize}
		\item If ${\sim}$ is the LRC of $\FF_L$ and $xa \mathrel{\not{\sim}_L} y$:\\
		      There exist $u \in \Sigma^*, v \in \Sigma^+$ with $xauv^\omega \in L$ if and only if $yuv^\omega \notin L$, i.e.~$uv^\omega$ separates $xa$ and $y$ in ${\sim}_L$. We pick the length-lexicographically minimal such $u$ and $v$
		      and add the words $xauv^\omega$ and $yuv^\omega$ to $W$.
		\item If ${\simeq}_c$ is the PRC for some class $c \in [{\sim}_L]$ and $xa \mathrel{\not{\simeq}_c} y$:\\
		      Let $u = \minrep{c}$. If not $uxa \mathrel{{\sim}_L} uy$, then we proceed as for the LRC. Otherwise, there exists some $z \in \Sigma^*$ with $uxaz \mathrel{{\sim}_L} u$ and $(u(xaz)^\omega \in L \iff u(yz)^\omega \notin L)$.
		      Again, we choose the length-lexicographically minimal such $z$ and add $u(xaz)^\omega$ and $u(yz)^\omega$ to $W$.
	\end{itemize}

	For computing the coloring, \dpainf considers all words in $R_{c,\sigma}$ for $\sigma \in \{+,-\}$.
	To ensure that the second step of \dpainf terminates with the precise coloring $\family\pi^{{\sim}_L}$ for $(L, {\sim}_L)$ we add to $W$ words which guarantee that all idempotent classes of ${\simeq}_c$ are in the infinity set of some word in $R_c$.
	Let $c \in [{\sim}_L]$ be a class and $u \in c$.
	From each idempotent class of ${\simeq}_c$, we pick the length-lexicographically least representative $x$ and add $ux^\omega$ to $W$.
	Note that since $x$ is idempotent in ${\simeq}_c$, it must also be in $E_c$ and hence $x^\omega$ will be in the set $R_c$ that is computed in the beginning of step 2.
	This guarantees that if the syntactic FORC is learned in the first step, every idempotent class of ${\simeq}_c$ obtains a correct classification in step 2 of \dpainf.
	As the subsequent computation of the mapping mirrors that of $\kappa_c$ from \cref{forcs}, it is guaranteed that the family of Mealy machines $\overline\MM$ returned by step 2 computes the precise FWPM $\family{\pi^{{\sim}_L}}$.

	Let $\bow\AA := \bow\AA(\overline\MM)$ and $\BB$ be as constructed in step 3.
	Further, let $S = (S_+, S_-) = (W \cap L, W \setminus L)$ be the obtained from the set $W$ constructed so far.
	Since $\overline\MM$ computes the precise FWPM $\family{\pi^{{\sim}_L}}$, it follows from \cref{lem:bowacomputesjoin} that $\bow\AA$ computes the priority mapping of $\AA_L$.
	By definition, $\BB$ behaves like $\bow\AA$ on prefixes from $S$, which guarantees that $\BB$ must be consistent with $S$ and the algorithm does not default to the fallback DPA.
	However, in the last step, there might be some $\op{output}(u)$ queries posed by \mealyal, which are (falsely) answered with $0$ purely because $u$ is not a prefix of $S$.
	Thus, to ensure that all queries are answered correctly, we simulate an execution of \mealyal with a teacher for $\bow\AA$ and whenever \mealyal poses an $\op{output}(u)$ query for a word that is not yet a prefix of $W$, we add $u^\omega$ to $W$.
	Additionally, for each equivalence query $\op{equiv}(\HH)$ with $L(\HH) \neq L$, we add the least counterexample to $W$ that witnesses the inequality.
	The overall construction of $S^L$ is guaranteed to terminate, because all answers given to \mealyal are consistent with $\AA_L$, which by the properties of \mealyal guarantees that it terminates at the latest with the DPA $\AA_L$.

	In the end, $S^L = (S_+, S_-) = (W \cap L, W \setminus L)$ is the characteristic sample for $L$.
	Our construction ensures that if called on a sample which is consistent with $L$ and contains $S^L$, \dpainf infers the syntactic FORC of $L$ in step~1, and then computes the precise FWPM $\family{\pi^{{\sim}_L}}$ in step 2. The output queries of \mealyal are all on prefixes of $S^L$, and as long as \mealyal does not propose a hypothesis that is consistent with $L$, the sample contains an example witnessing the difference. Hence, \dpainf terminates with a DPA $\AA$ for $L$.

	For an upper bound on the size of $S^L$, we consider the words which are added in each step.
	The number of insertions that have to be prevented in step 1 and the number of idempotent classes that need to be considered in step two are clearly polynomial in the syntactic FORC $\FF_L$. The shortest words for witnessing the inequivalences that prevent the insertions are polynomial in the size of $\FF_L$.
	In step four, the number of added words (because of output or equivalence queries) and their lengths are polynomial in the size of the precise DPA $\AA_L$ because the hypotheses are growing and bounded by the size of $\AA_L$, and the algorithm runs in time polynomial in the resulting Mealy machine and the length of the longest counterexample.
	In summary, the size of $S^L$ is polynomial in the maximum of $|\FF_L|$ and $|\AA_L|$.
\end{proof}

By combining \cref{dpainfcomplete} with \cref{FORCsize} and \cref{boundpreciseDPA}, we can recover and extend the currently known results on passive learning of deterministic $\omega$-automata from polynomial data.
For that purpose, we consider the subclasses $\irc$ and $\polyclass$ of the $\omega$-regular languages. The first one contains all languages $L$ that can be accepted by a DPA that has exactly one state for each ${\sim}_L$-class. Polynomial time passive learners that can learn each language from this class in the limit from polynomial data are presented in  \cite{AngluinFS20,BohnL21}. The class $\polyclass$ contains all languages $L$ that can be recognized by a DPA with priorities in $\{0, \ldots, k-1\}$ and at most $d$ many $c$-states for each ${\sim}_L$ class $c$. A polynomial time passive learner that, for fixed $d$, can learn each language in \polyclassBuchi in the limit from polynomial data is presented in  \cite{BohnL22}

The languages in $\irc$ are those for which $d=1$ in \cref{FORCsize} and \cref{boundpreciseDPA}. In this case, the term in \cref{boundpreciseDPA} reduces to $m$ and the one in \cref{FORCsize} to $mk$. If $d$ and $k$ are both fixed, both expressions become linear in $m$. This implies the following.

\begin{restatable}{corollary}{rstfixedparams}\label{fixedparams}
	The algorithm \dpainf can learn a DPA for every language in \irc from polynomial data.
	Furthermore, there is a fixed polynomial $g$ such that for every $k$ and $d$, \dpainf can learn a DPA for every language in $\polyclass$ in the limit from polynomial data with characteristic samples of size $\mathcal{O}(g)$.
\end{restatable}

\section{Conclusion} \label{conclusion}

We have presented a passive learner for deterministic parity automata that runs in polynomial time and can learn a DPA for every regular $\omega$-language in the limit. Our upper bound for the size of complete samples is, in general, exponential in the size of a minimal DPA for the language. However, for fixed number of priorities and fixed maximal number of pairwise language equivalent states, this bound becomes polynomial. The learning algorithm is based on the precise DPA of a language that we introduced in this paper, and that can be constructed from the syntactic FORC of the language.

We see two natural main directions of future research based on the results presented here.
First, we proposed a basic version of the algorithm that is complete for the class of regular $\omega$-languages and runs in polynomial time. But there are many parts of the algorithm that allow for optimizations and variations without losing these properties. These variations should then be implemented and compared in an empirical study.
For example, all the progress congruences of the FORC are learned independently, while there are strong dependencies between the progress congruences of the syntactic FORC. One can explore techniques for learning these congruences while respecting mutual dependencies. And one can also try other types of passive learning algorithms from finite automata adapted to learning FORCs. Also, variations of the active learner used in the last step of the algorithm can have an impact on the running time and the result produced by the overall procedure.

Second, the precise DPA for a language deserves further study.
An understanding which structural properties of a language can cause the precise DPA to be much larger than a minimal DPA for the language might give insights to minimization problems for DPAs.
Furthermore, our construction of the precise DPA from the syntactic FORC has similarities with a construction that starts from the syntactic semigroup (see specifically Lemma 22 in \cite{Colcombet11}). And the construction of the precise FWPM from the syntactic FORC using idempotent classes suggests a variation of FDFAs with ``idempotent acceptance'', which could lead to smaller representations of $\omega$-languages by FDFAs.
There is also a connection to the canonical representation by good-for-games automata from \cite[Definition 4]{EhlersS22}, in the sense that the precise DPA computes the natural color in the limit.
Finally, one can show that the class of precise DPAs for a language subsumes the class of normalized DPAs in the sense that a DPA $\AA$ is normalized and minimal as a Mealy machine if it is the precise DPA for $(L(\AA), {\sim}_\AA)$.
So it seems that studying precise DPAs and their construction in more detail can lead to further insights on connections between representations of $\omega$-languages.

\printbibliography
\end{document}